\documentclass[a4paper,onecolumn,superscriptaddress,10pt,floatfix]{quantumarticle}
\pdfoutput=1

 \usepackage[utf8]{inputenc}

\usepackage{amsmath,amsfonts,amssymb,amsthm,graphicx,bbm,enumitem}

\usepackage{xcolor}
\usepackage{mathtools}
\usepackage{graphicx} 
\usepackage{float}    
\usepackage{verbatim} 
\usepackage{amsmath}  
\usepackage{amssymb}  
\usepackage{mathtools}
\usepackage{listings} 
\usepackage{eufrak}   
\usepackage[pdftex]{hyperref}           

\newtheorem{thm}{Theorem}
\newtheorem{lm}[thm]{Lemma}
\newtheorem{cor}[thm]{Corollary}
\newtheorem{com}[thm]{Comment}

\newtheorem{innercustomthm}{Theorem}
\newenvironment{customthm}[1]
  {\renewcommand\theinnercustomthm{#1}\innercustomthm}
  {\endinnercustomthm}

\newcommand{\R}{\mathbb{R}}
\newcommand{\N}{\mathbb{N}}
\newcommand{\C}{\mathbb{C}}

\newcommand{\wor}{\ast}

\newcommand{\maj}{m}
\newcommand{\fer}{f}

\DeclareMathOperator{\tr}{tr}

\DeclareMathOperator*{\EE}{\mathbb{E}}

\DeclareMathOperator{\supp}{supp}

\DeclareMathOperator*{\argmin}{argmin}

\newcommand{\fock}[1]{\mathcal{F}_{#1}}
\newcommand{\corr}{\Gamma}

\newcommand{\hil}{\mathcal{H}}

\newcommand{\ketbra}[2]{|#1\rangle\langle#2|}

\newcommand{\spdim}{D}
\newcommand{\es}{|E|}
\newcommand{\eis}[1]{|E_{#1}|}

\newcommand{\id}{\mathbbm{1}}

\newcommand{\locnummodes}{p}	
\newcommand{\syssize}{N}	
\definecolor{jens}{rgb}{0,.8,.5}

\begin{document}

\title{Extendibility of fermionic states and rigorous ground state approximations of interacting fermionic systems}

\author{Christian Krumnow}
\affiliation{Center for Biomedical Image and Information Processing, University of Applied Sciences (HTW) Berlin, 12459 Berlin, Germany}
\affiliation{Dahlem Center for Complex Quantum Systems, Freie Universit{\"a}t Berlin, 14195 Berlin, Germany}

\author{Zolt\'an Zimbor\'as}
\affiliation{HUN-REN Wigner Research Centre for Physics, P.O. Box 49, H-1525 Budapest, Hungary}

\author{Jens Eisert} 
\affiliation{Dahlem Center for Complex Quantum Systems, Freie Universit{\"a}t Berlin, 14195 Berlin, Germany}
\affiliation{Helmholtz-Zentrum Berlin f{\"u}r Materialien und Energie, 14109 Berlin, Germany}

\begin{abstract}
Solving interacting fermionic quantum many-body problems as they are ubiquitous in quantum chemistry and materials science is a central task of theoretical and numerical physics, a task that can commonly only be addressed in the sense of providing approximations of ground states. For this reason, it is important to have tools at hand to assess how well simple ansatzes would fare. In this work, we provide rigorous guarantees on how well fermionic Gaussian product states can approximate the true ground state, given a weighted interaction graph capturing the interaction pattern of the systems. Our result can be on the one hand seen as a extendibility result of fermionic quantum states: It says in what ways fermionic correlations can be distributed. On the other hand, this is a non-symmetric de-Finetti theorem for fermions, as the direct fermionic analog of a theorem due to Brandao and Harrow. We compare the findings with the distinctly different situation of distinguishable finite-dimensional quantum systems, comment on the approximation of ground states with Gaussian states and elaborate on the connection to the no low-energy trivial state conjecture.
\end{abstract}

\maketitle


\section{Introduction}

\noindent
Interacting quantum systems constitute among the most intricate problems in modern physics and materials science. In particular in their fermionic reading, they are ubiquitous as models for systems in quantum chemistry \cite{CramerQuantumChemistry,Lilienfeld} and for 
realistic quantum materials \cite{QuantumMaterials}. 
Studies of properties of quantum materials are common in particular in modern numerical approaches: In fact, more than one million CPU hours are spent on supercomputers every day to compute properties of interacting quantum matter \cite{draxl_scheffler_2018}.
In situations in which interactions matter -- so when second quantized Hamiltonians can no longer be reasonably treated as being quadratic polynomials in fermionic operators and for which quartic terms substantially contribute -- one faces to deal with challenges that come along with exponentially large Hilbert spaces. 

For this reason, and motivated by the overwhelming importance of interacting fermionic problems, a wealth of approximation schemes has been introduced. The presumably most common of such approximations is the \emph{Hartree-Fock approximation} that basically approximates ground states with Gaussian fermionic states, so ground states of non-interacting systems featuring quadratic fermionic Hamiltonians. \emph{Variational approaches} approximate ground state energies from above. Indeed, they are much used in physics, in particular in the reading of 
\emph{tensor network approaches}
\cite{Orus-AnnPhys-2014,ReiherChemistry,TroyerChemistry,Krumnow2016,Pfeffer} 
-- that make use of ground states featuring little entanglement in a precise sense \cite{AreaReview} -- but also in quantum \emph{variational eigensolvers} \cite{McClean_2016}.
\emph{Hierarchies of semi-definite relaxations} approximate ground state energies from below \cite{Mazziotti2002}. Mean-field
approaches come in several specific (and distinctly different) readings: But often, in particular when local systems feature many neighbours in a lattice, even product or \emph{Gutzwiller} ansatz states can provide good ground state approximations
\cite{PhysRevLett.59.121}. 
The intuition why such a mean-field approach should be good stems is guided by insights provided so-called de-Finetti theorems \cite{Christandl2007,Krumnow2017}: Roughly speaking, for symmetric quantum many-body systems where each constituent interacts with many neighbors, product states are expected to provide good ground state approximations.

But at the end of the day, crucial questions that remain open in any of those approaches are  of the following type: To what extent can one guarantee Gaussian fermionic product states to provide a reasonable approximation to ground states of fermionic Hamiltonians? Guided by an intuition of what is often referred to as \emph{monogamy of entanglement} \cite{Monogamy} and an intuition guided by the de-Finetti theorem, can one expect ground states of systems in which local constituents have many neighbors to be well approximated by  product states featuring no entanglement whatsoever? Can such an intuition even made precise in the absence of any permutation symmetry, but for realistic interacting fermionic systems?

In this work, we contribute substantial steps towards establishing rigorous bounds to the approximation error when approximating ground states of interacting fermionic quantum systems with Gaussian pure product states. In this sense, we establish guarantees to which a notion of a monogamy of entanglement can be mathematically precisely stated. In our arguments, we do not rely on permutational symmetry -- in fact, symmetry does not matter 
whatsoever: In the presence of such a symmetry,  variants of 
de-Finetti theorems can readily be seen as providing estimates of the desired type
\cite{Christandl2007,Krumnow2017}. In the 
realistic setting of an absence of a 
permutational symmetry, such theorems are technically much harder to come by, which is yet covered here. 
To achieve our results, we build on and extend the machinery developed in Ref.~\cite{Brandao2016} on rigorous product-state 
approximations to quantum ground states for spin models: In fact, our result can be seen as an immediate analog of a theorem due to Brandao and
Harrow. Indeed, since this work contains some minor technical mistakes,
we redo and correct this work, leading so somewhat different conclusions which we
state explicitly. 

Yet, by no means is this an immediate 
corollary of those results: For the fermionic nature of the system under considerations, the logic of those arguments has to be substantially altered. In a way, our
result can be seen as a result on the \emph{sharability}: We show that fermionic correlations are monogamous in the sense that if a mode shares lots of anti-symmetric correlations with an other mode this will limit the amount of correlations build up to further modes. Another way of seeing the results is that it constitutes a fermionic de-Finetti theorem without permutation symmetry: It is still true that system being coupled to many neighbors will have ground states that are close to products. In practice, the results allow to judge from the Hamiltonian alone to what extent a
Gaussian product state is a good approximation: The upshot of our arguments is that the weighted interaction graph alone, capturing couplings in second quantized Hamiltonians, already dictates how well a Gaussian product approximation will work. 

In our endeavours, we relate to a number of recent findings that are technically different but related in mindset. Our 
main aim in our work is to completely and comprehensively characterize the extendibility of fermionic states. The situation we find is distinctly different from that of finite-dimensional quantum systems, which we compare it with. This distinctly different situation between fermions and spins is reflecting other comparable dissimilarities. For example, it is known that low-degree fermionic monomials have a very different commutation structure than low-weight Pauli operators \cite{AnschuetzNonGlassy}. 
The commutation index captures this difference and quantifies the subtle distinction in the physics of local spin systems and local fermionic systems. 
Building on these insights, it has been suggested that low-temperature strongly interacting fermions, unlike spins, belong in a classically nontrivial yet quantumly easy phase. The findings also relate so studies exploring to what extent one can expect Gaussian fermionic states to provide good ground state approximations altogether 
\cite{HastingsFermions}, studied at hand of the \emph{Sachdev--Ye--Kitaev} (SYK) model. Also, recently,
a fermionic reading of the 
\emph{no low-energy trivial state}
(NLTS) conjecture \cite{NLTS} 
has been settled \cite{PhysRevA.109.052431}, as has been 
before for spin models in seminal work building on notions of quantum error correction \cite{FermionicNLTS}: For a Hamiltonian with no low-energy trivial states, a fermionic Gaussian product state will also not be a good ground state
approximation. These various links are being 
elaborated on at the end of this manuscript.

\section{Setting and definitions}
Throughout this work, we will discuss finite 
lattice systems, the vertices of which are associated
with distinguishable particles as well as fermions. In order to capture this, let $V$ be a discrete and finite set of $\syssize=|V|$ sites and $(V,E)$ a graph that reflects the lattice. The vertices of the graph are associated with physical fermionic degrees of freedom.
For each vertex or site $i\in V$, we define the set of neighbors $E_i = \{j:(i,j)\in E\}$ with $\sum_{i\in V}\eis{i}=2\es$.
Specifically, they are modes in a second quantized picture, obtained from first quantization
upon picking a specific set of orbitals. 
That is to say, each vertex or 
site $j\in V$ is
associated with $\locnummodes$ fermionic
modes or orbitals, being
equipped with the full Fock space $\fock{\locnummodes\syssize}$. 
The corresponding 
Hermitian \emph{Majorana operators} $\maj_{j}^{2\alpha-1}$ 
and $\maj_{j}^{2\alpha}$ for $\alpha\in[p]$ satisfy
\begin{align}
 \{\maj_j^{\alpha},\maj_k^\beta\} = 2\delta_{j,k}\delta_{\alpha,\beta}\id.
\end{align}
In terms of the more familiar fermion creation and 
annihilation operators, one has
$m_j^{2\alpha-1} = ((f_j^\alpha)^\dagger 
+ f_j^\alpha)$ and
$m_j^{2\alpha} = i ((f_j^\alpha)^\dagger 
- f_j^\alpha)$.
For site $j\in V$, define the \emph{local parity operator} as
\begin{align}
 P_j := (-i)^\locnummodes\prod\limits_{\alpha=1}^\locnummodes\maj_j^{2\alpha-1}\maj_j^{2\alpha}.
\end{align}
Further, we define for $j\in V$ the maps $\Xi_j:\mathcal{B}(\fock{\locnummodes\syssize})\rightarrow \mathcal{B}(\fock{\locnummodes\syssize})$ and $\Xi:\mathcal{B}(\fock{\locnummodes\syssize})\rightarrow \mathcal{B}(\fock{\locnummodes\syssize})$ via
\begin{align}
 \Xi_j(X) &:= \frac{1}{2}(X+P_jXP_j), \\
 \Xi &:= \Xi_1\circ\dots \circ \Xi_\syssize.
\end{align}
As the Majorana operators generate the operator algebra on $\fock{\locnummodes\syssize}$, we can write any observable $A\in\mathcal{B}(\fock{\locnummodes\syssize})$ in terms of the Majorana operators as
\begin{align}
A = \sum\limits_{I\subset V\times[2\locnummodes]}c_I\prod\limits_{(j,\alpha)\in I} \maj_j^\alpha. \label{eq:DefObsExpansion}
\end{align}
We call an observable \emph{totally even} if $\Xi(A)=A$ which means that its expansion according to \eqref{eq:DefObsExpansion} consists of only terms where for all $j\in V$, we have that $|\{\alpha:(j,\alpha\in I)\}|$ is even (meaning that in each term an even number of Majorana operators is associated to any site). The projection operators above and their products are examples for 
totally even operators.
Similarly we denote an observable as \emph{totally odd} if for all $j\in V$, we have that $|\{\alpha:(j,\alpha\in I)\}|$ is odd or zero. The individual Majorana operators are hence totally odd.

Note that when restricting the domain of $\Xi$ to the states $\mathcal{D}(\fock{\locnummodes\syssize})$ it constitutes a quantum channel mapping $\rho\mapsto \Xi(\rho)=\sigma_\rho$ which is totally even 
and related to $\rho$ via
\begin{equation}
 \tr(\sigma_\rho A) = \begin{cases} \tr(A \rho)&\text{if $A$ is totally even,}\\0&\text{else}.\end{cases}
\end{equation}
In the case of distinguishable particles, we associate to each site $i$ of $(V,E)$ a local $d$-dimensional Hilbert space $\hil_i = \C^d$.
For any quantum state $\rho$, either for fermionic or  distinguishable particles, and $i,j\in V$, we define $\rho^{i,j}$ to be the reduction of $\rho$ to sites $i$ and $j$.

For the ease of notation, we define for a set $S$ and a functions $f$ supported on $S$ the uniform average
\begin{equation}
 \EE\limits_{s\in S}f(s) = \frac{1}{|S|}\sum\limits_{s\in S}f(s).
\end{equation}
Furthermore, given a discrete random variable $x$, some probability distribution $\mu$ of it and a function $f$ depending on $x$ we write
\begin{equation}
 \EE\limits_{x\sim \mu}f(x) = \sum\limits_{x}\mu(x)f(x).
\end{equation}
For a probability distribution $\mu$ over $[n]$ and $k<|\supp(\mu)|$ we denote with $\mu^{\wor k}$ the distribution of drawing $k$ times without replacement which is explicitly given as
\begin{equation}
 \mu^{\wor m}(i_1,\dots,i_k) = \begin{cases}
                           0 &\text{if }i_1,\dots,i_k\text{ are not all distinct,} \\
                           \frac{\mu(i_1)\dots\mu(i_k)}{(1-\mu(i_1))\dots(1-\mu_{i_1}-\dots-\mu(i_{k-1}))} &\text{otherwise.}
                          \end{cases}\label{defSWOR}
\end{equation}
It can be easily checked from the implicit definition as well as the explicit form stated above that given a tuple $(i_1,\dots,i_k,i_{k+1})$, summing over all possible last entries 
$i_{k+1}$ yields the probability of drawing $(i_1,\dots,i_k)$, 
i.e.,
\begin{equation}
 \sum\limits_{i_{k+1}\neq i_1,\dots,i_k}\mu^{\wor k+1}(i_1,\dots,i_k,i_{k+1}) = \mu^{\wor k}(i_1,\dots,i_k).\label{eqMarginalProp}
\end{equation}


\section{Monogamy of fermionic two-site correlations}\label{sec:Monogamy of fermionic two site correlation}
In order to relate a general fermionic quantum state to a mode product state we start out by bounding
the anti-symmetric two-site correlation a state $\rho$ might contain for any site $i$ with its respective neighbors $E_i$. This step is crucial for understanding the possible correlation pattern quantum
mechanical states can feature. For this, we define for $i\in V$
\begin{align}
 \corr_i = \sum_{j\in E_i} \|\rho^{i,j}-\sigma_\rho^{i,j}\|_1.
\end{align}
Here, $\|.\|_1$ denotes the 1-norm or trace-norm that has a 
clear operational interpretation in terms of the
distinguishability of two quantum states.
Recall here that $\sigma_\rho$ captures the (totally) 
even part of $\rho$ only. With this, $\sigma_\rho$ is insensitive to any non-local fermionic aspect of $\rho$, in the sense that for any observable consisting of two anti-commuting operators $A_i$ and $B_j$ supported on $i\neq j\in V$ with $A_iB_j = -B_jA_i$ we have $\tr(A_iB_j\sigma_\rho) = 0$. The covariance matrix of $\sigma_\rho$ with entries $i\tr(\sigma_\rho[\maj^\alpha_j,\maj^\beta_k])/2$ for $j,k\in V$, $\alpha,\beta\in [2p]$ decomposes into an on-site block structure and vanishes in the case of $p=1$, i.e., a single mode per site.
The difference $\rho-\sigma_\rho$ is then the part of $\rho$ which contains the non-local fermionic (or anti-symmetric) correlations.

The mechanism which restricts these correlations is at least in special cases well known but has been so far not exploited in the way presented here, to our best knowledge. 
The basic insight needed is that the norm of linear combinations of odd operators scales non-extensively. 
This result can be seen as a generalization of the fact that the norm of $\sum_{i=1}^N \maj_i^\alpha$ scales as $\sqrt{N}$ which can be seen from realizing that $\sum_{i=1}^N \maj_i^\alpha/\sqrt{N}$ is a valid Majorana mode operator with unit norm.
The generalization needed is stated in Lemma~\ref{lm:norm_of_linear_ferm_combinations} which allows us to prove that fermionic two site correlations behave monogamous in the following sense:

\begin{thm}[Monogamy of fermionic two-site correlations]\label{thm:monogamy}
 For any state $\rho\in\mathcal{D}(\fock{\locnummodes\syssize})$ and site $i\in V$, we find that 
 \begin{equation}
 \EE\limits_{j\in E_i} \|\rho^{i,j}-\sigma_\rho^{i,j}\|_1 \leq \frac{2^{4\locnummodes}}{4\sqrt{\eis{i}}}.
\end{equation}
\end{thm}
Theorem~\ref{thm:monogamy} states that the amount of odd correlations a site shares with its neighbors on average scales down in the number of neighbors present.
In order to prove this theorem, we first need the following lemma on the norm of totally odd operators.

\begin{lm}[Norm of linear combinations of odd operators]\label{lm:norm_of_linear_ferm_combinations}
 Let $V_1,V_2\subset V$ with $V_1\cap V_2 = \emptyset$ and $J,K\subset [2\locnummodes]$ with $|J|, |K|$ odd. 
 Given for any $j\in V_1$ a weight $c_j\in \R$, we find
\begin{align}
 \left\| \sum\limits_{j\in V_1} c_j \prod\limits_{\alpha\in J}\maj_j^\alpha\right\| = \sqrt{\sum\limits_{j\in V_1}c_j^2}.
\end{align}
In addition, 
we have
\begin{align}
 \left\|\sum\limits_{j\in V_1}\sum\limits_{k\in V_2} \prod\limits_{\alpha\in J}\maj_j^\alpha\prod\limits_{\beta\in J}\maj_k^\beta\right\| = \sqrt{|V_1||V_2|}.
\end{align}
\begin{proof}
 Both equations are proven by considering the square of the operator of interest.
 For $j\in V$ and $J\subset [2\locnummodes]$ define $\prod_{\alpha\in J}\maj_j^\alpha = \maj_j^J$.
 Furthermore, define
 \begin{equation}
  O_1 := \sigma_{|J|}\sum\limits_{j\in V_1 }c_j \maj_j^J\quad,\qquad\qquad O_2 := \sigma_{|J|+|K|}\sum\limits_{j\in V_1 }\sum\limits_{k\in V_2} \maj_j^J\maj_k^K\quad,
 \end{equation}
 where the pre-factor $\sigma_{n} := \sqrt{(-1)^{(n-1)/2}}$ ensures that $O_1$ and $O_2$ are Hermitian with 
 $\|O_1\| = \|\sum\limits_{j\in V_1} c_j \maj_j^J\|$ and  \begin{equation}
 \|O_2\| = \|\sum\limits_{j\in V_1 }\sum\limits_{k\in V_2} \maj_j^J\maj_k^K\|.
  \end{equation}
 We then find that 
 \begin{align}
  O_1^2 &=  \sigma_{|J|}^2\sum\limits_{j\in V_1}c_jc_j \maj_j^J\maj_j^J + \sigma_{|J|}^2\sum\limits_{j,k\in V_1|j<k}c_jc_k \maj_j^J\maj_k^J + \sigma_{|J|}^2\sum\limits_{j,k\in V_1|j>k}c_jc_k \maj_j^J\maj_k^J .
 \end{align}
 For the first term, we have that $\maj_j^J\maj_j^J = \sigma_{|J|}^2\id$. In the second and third term, the contributions $c_jc_k \maj_j^J\maj_k^J$ and $c_kc_j \maj_k^J\maj_j^J$ cancel each other for every $j<k$ due to the anti-commutation relations. Hence, we are in the position to conclude
\begin{align}
  O_1^2 &= \sigma_{|J|}^4\sum\limits_{j\in V_1}c_jc_j \ \id = \sum\limits_{j\in V_1}c_jc_j\ \id.
 \end{align}
Similarly, we obtain 
\begin{align}
 O_2^2 =& \sigma_{|J|+|K|}^2\sum\limits_{j\in V_1,k \in V_2} \maj_j^J\maj_k^K\maj_j^J\maj_k^K + \sigma_{|J|+|K|}^2 \sum\limits_{j,x\in V_1}\sum\limits_{k,y\in V_2| k<y}\maj_j^J\maj_k^K\maj_x^J\maj_y^K \nonumber\\
 &+\sigma_{|J|+|K|}^2 \sum\limits_{j,x\in V_1}\sum\limits_{k,y\in V_2| k>y}\maj_j^J\maj_k^K\maj_x^J\maj_y^K + \sigma_{|J|+|K|}^2 \sum\limits_{j,x\in V_1| j<x}\sum\limits_{k\in V_2}\maj_j^J\maj_k^K\maj_x^J\maj_k^K \nonumber\\
 &+\sigma_{|J|+|K|}^2 \sum\limits_{j,x\in V_1| j>x}\sum\limits_{k\in V_2}\maj_j^J\maj_k^K\maj_x^J\maj_k^K.
\end{align}
Again, we find that the second and third term cancel each other as well as the fourth and fifth. For the first term we obtain that $\maj_j^J\maj_k^K\maj_j^J\maj_k^K = \sigma_{|J|+|K|}^2\id$ which yields
\begin{align}
 O_2^2 =& |V_1||V_2|\ \id.
\end{align}
For any Hermitian operator $O$, 
we have that $O^2 = c\,\id$ implies that $O$ can have eigenvalues $\pm\sqrt{c}$ only such that $\|O\| = \sqrt{c}$ from which the claim follows.
\end{proof}
\end{lm}

This allows us to prove Theorem~\ref{thm:monogamy} as follows.
\begin{proof}
 First, we rewrite the one-norm in terms of expectation values
\begin{align}
 \corr_i &= \sum_{j\in E_i} \|\rho^{i,j}-\sigma_\rho^{i,j}\|_1\\
 	&=\sum_{j\in E_i} \sup\limits_{A^{i,j}:\|A^{i,j}\|=1}\tr(A^{i,j}[\rho^{i,j}-\sigma_\rho^{i,j}]),\nonumber
\end{align}
where the supremum runs over normalized observables supported on the sites $i,j$ only.
Any observable $A$ can be split into $A=A_\text{even}+A_\text{odd}$, so into an totally even part $A_\text{even}$ and a remaining part $A_\text{odd}$ containing no totally even term with $\|A_\text{even}\|\leq \|A\|$ and $\|A_\text{odd}\|\leq \|A\|$ 
(see, e.g., Ref.~\cite{Krumnow2017}).
By definition $\rho$ and $\sigma_\rho$ agree for totally even observables such that we can restrict the supremum to non-totally even operators which are automatically totally odd for two sites.
Hence, we find 
\begin{align}
 \corr_i = \sum_{j\in E_i} \sup\limits_{\substack{A^{i,j}:\|A^{i,j}\|=1\\\text{totally odd}}}\tr(A^{i,j}[\rho^{i,j}-\sigma_\rho^{i,j}])=\sum_{j\in E_i} \sup\limits_{\substack{A^{i,j}:\|A^{i,j}\|=1\\\text{totally odd}}}\tr(A^{i,j}\rho^{i,j}).
\end{align}
Rewriting the expression yields
\begin{align}
 \corr_i = \sup\limits_{\substack{\{A^{i,j}|j\in E_i\}:\|A^{i,j}\|=1\\\text{totally odd}}}\tr\left(\sum_{j\in E_i}A^{i,j}\rho^{i,j}\right).
\end{align}
Any totally odd two site observable $A$ can be expanded as 
\begin{align}
 A = \sum\limits_{\substack{J\subset[2\locnummodes]\\|J|\text{ odd}}}\sum\limits_{\substack{L\subset[2\locnummodes]\\|L|\text{ odd}}}c_{\{i\}\times J\cup \{j\}\times L}\prod\limits_{\alpha\in J}\maj_i^\alpha\prod\limits_{\beta\in L}\maj_j^\beta.
\end{align}
Interchanging the order of summation yields then
\begin{align}
 \sum\limits_{j\in E_i} A^{i,j}  = \sum\limits_{\substack{J\subset[2\locnummodes]\\|J|\text{ odd}}}\sum\limits_{\substack{L\subset[2\locnummodes]\\|L|\text{ odd}}}\sum\limits_{j\in E_i}c_{\{i\}\times J\cup \{j\}\times L}\prod\limits_{\alpha\in J}\maj_i^\alpha\prod\limits_{\beta\in L}\maj_j^\beta.
\end{align}
Inserting this in the expression above, we then find 
\begin{align}
 \corr_i &=  \sup\limits_{\substack{\{A^{i,j}|j\in E_i\}:\|A^{i,j}\|\leq 1,\\\text{totally odd}}} \sum\limits_{\substack{J\subset[2\locnummodes]\\|J|\text{ odd}}}\sum\limits_{\substack{L\subset[2\locnummodes]\\|L|\text{ odd}}}\tr\left(\left[\sum_{j\in E_i} c_{\{i\}\times J\cup \{j\}\times L}\prod\limits_{\alpha\in J}\maj_i^\alpha\prod\limits_{\beta\in L}\maj_j^\beta\right]\rho^{i,j}\right).
\end{align}
For any set of observables $\{A^{i,j}| j\in E_i\}$, we can bound
\begin{align}
\sum\limits_{\substack{J\subset[2\locnummodes]\\|J|\text{ odd}}}\sum\limits_{\substack{L\subset[2\locnummodes]\\|L|\text{ odd}}}&\tr\left(\left[\sum_{j\in E_i} c_{\{i\}\times J\cup \{j\}\times L}\prod\limits_{\alpha\in J}\maj_i^\alpha\prod\limits_{\beta\in L}\maj_j^\beta\right]\rho^{i,j}\right)\nonumber\\
 	&\leq \sum\limits_{\substack{J\subset[2\locnummodes]\\|J|\text{ odd}}}\sum\limits_{\substack{L\subset[2\locnummodes]\\|L|\text{ odd}}} \left\|\prod\limits_{\alpha\in J}\maj_i^\alpha\right\|
 	\left\|\sum_{j\in E_i} c_{\{i\}\times J\cup \{j\}\times L}\prod\limits_{\beta\in L}\maj_j^\beta\right\|\\
 	&\leq \sum\limits_{\substack{J\subset[2\locnummodes]\\|J|\text{ odd}}}\sum\limits_{\substack{L\subset[2\locnummodes]\\|L|\text{ odd}}} \sqrt{\eis{i}} \leq 2^{4\locnummodes}\sqrt{\eis{i}}/4,\nonumber
\end{align}
where in the second line we have used that $\|A^{i,j}\|=1$, which in turns implies that $|c_{\{i\}\times J\cup \{j\}\times L}|\leq 1$ for all coefficients. We have also  employed  Lemma~\ref{lm:norm_of_linear_ferm_combinations}.
This yields 
\begin{align}
 \corr_i &= \leq 2^{4\locnummodes}\sqrt{\eis{i}}/4.
\end{align}
In expectation, we obtain 
for any $i\in V$ the 
claim as
\begin{equation}
 \EE\limits_{j\in E_i} \|\rho^{i,j}-\sigma_\rho^{i,j}\|_1 = \frac{1}{\eis{i}}\corr_i \leq \frac{2^{4\locnummodes}}{4\sqrt{\eis{i}}}.
\end{equation}
\end{proof}

We immediately obtain the following corollary for the expectation value over the full edge set of the graph.

\begin{cor}[Expectation value over the edge set]\label{cor:full_edge1}
Summing Theorem~\ref{thm:monogamy} over all nodes yields for the average over the full edge set
\begin{equation}
 \EE\limits_{(i,j)\in E}\|\rho^{i,j}-\sigma_\rho^{i,j}\|_1 = \frac{1}{2|E|}\sum\limits_{i\in V}\Gamma_i \leq \sum\limits_{i\in V}\frac{ 2^{4\locnummodes}\sqrt{\eis{i}}}{8|E|}.
 \end{equation}
 \end{cor}
 However, depending on the graph, this bound might not be optimal or in instances even trivial. The reason for this is that in Corollary~\ref{cor:full_edge1} we sum over each edge twice by considering every node. A graph with many nodes of low degree, however, in the extreme case a star-like graph, will then yield only a trivial bound. We can overcome this limitation by realizing that for a given $(V,E)$ we do not need to consider all nodes in order to have any edge to be adjacent to at least one chosen node. In general, 
 we can conclude the following tighter bound to hold true.
 
\begin{cor}[Tighter bound]\label{cor:full_edge2}
Let $V^\prime$ denote any vertex cover of $(V,E)$ we obtain that 
\begin{equation}
 \EE\limits_{(i,j)\in E}\|\rho^{i,j}-\sigma_\rho^{i,j}\|_1 \leq \frac{1}{|E|}\sum\limits_{i\in V^\prime}\Gamma_i \leq \sum\limits_{i\in V^\prime}\frac{2^{4\locnummodes}\sqrt{\eis{i}}}{4|E|}.
 \end{equation}
 \end{cor}
 Let us exemplify the obtained bounds for two more specific settings that give rise to particularly insightful examples of graphs.
 
 \begin{cor}[Bound for 
 $c$-regular and star graphs]\label{cor:full_edge_special}
 We obtain for $c$-regular graphs
 \begin{equation}
 \EE\limits_{(i,j)\in E}\|\rho^{i,j}-\sigma_\rho^{i,j}\|_1 \leq \frac{ 2^{4\locnummodes}}{8\sqrt{c} },
 \end{equation}
 and for a star which consists of one root and $\syssize-1$ many leaves
\begin{equation}
 \EE\limits_{(i,j)\in E}\|\rho^{i,j}-\sigma_\rho^{i,j}\|_1 \leq \frac{ 2^{4\locnummodes}}{4\sqrt{\syssize-1} }.
 \end{equation}
 \end{cor}

 Note, however, that the above combined bounds are in general not optimal as in the proof we leave out a further suppression of fermionic correlations which result from the correlations that a to site $i$ neighboring site $j$ shares its neighbors. One relevant case where the bounds can be tightened is given by the setting of an extendable state.
 
\begin{thm}[Correlations of extendable states in complete bipartite graphs]\label{thm:suppression_complete_bipartite}
Assume that $(V,E)$ is a complete bipartite graph, meaning that $V = V_1\cup V_2$ with $V_1\cap V_2 = \emptyset$ and $E = \{(a,b)|a\in V_1,b\in V_2\}$. Given a state $\omega\in\mathcal{D}(\fock{2\locnummodes})$ and a global state $\rho\in\mathcal{D}(\fock{\syssize\locnummodes})$ with $\rho^{i,j} = \omega$ for all $i\in V_1$, $j\in V_2$ then we have for any $i\in V_1$ and $j\in V_2$ that 
\begin{align}
\|\omega-\sigma_\omega\|_1 = \|\rho^{i,j}-\sigma_\rho^{i,j}\|_1 \leq \frac{2^{4\locnummodes}}{4}\frac{1}{\sqrt{|V_1||V_2|}}.
\end{align}
\end{thm}
\begin{proof}
 Consider the quantity
 \begin{equation}
  \Gamma := \sum\limits_{i\in V_1,j\in V_2}\|\rho^{i,j}-\sigma_\rho^{i,j}\|_1 .
 \end{equation}
 As in the proof of Theorem~\ref{thm:monogamy} above, we rewrite the one-norm in terms of an expectation value.
 Due to the special structure of the state we obtain that the local observables $A^{i,j}$ distinguishing $\rho^{i,j}$ and $\sigma^{i,j}_\rho$ will be the same for all $(i,j)\in V_1\times V_2$. Hence, we find that 
 \begin{align}
  \Gamma &\leq \sup\limits_{\substack{A:\|A\|\leq 1\\ \text{totally odd}}}\sum\limits_{\substack{J\subset [2\locnummodes]\\|J|\text{ odd}}}\sum\limits_{\substack{K\subset [2\locnummodes]\\|K|\text{ odd}}}c_{J,K}\tr\left(\sum\limits_{i\in V_1,j\in V_2}\prod\limits_{\alpha\in J}\maj_i^\alpha\prod\limits_{\beta\in K}\maj_j^\beta\rho\right) \\
  &\leq \sup\limits_{\substack{A:\|A\|\leq 1\\ \text{totally odd}}}\sum\limits_{\substack{J\subset [2\locnummodes]\\|J|\text{ odd}}}\sum\limits_{\substack{K\subset [2\locnummodes]\\|K|\text{ odd}}}c_{J,K}\left\|\sum\limits_{i\in V_1,j\in V_2}\prod\limits_{\alpha\in J}\maj_i^\alpha\prod\limits_{\beta\in K}\maj_j^\beta\right\|.
  \nonumber
 \end{align}
By Lemma~\ref{lm:norm_of_linear_ferm_combinations} and the fact that $|c_{J,K}|\leq 1$, 
we directly obtain 
\begin{equation}
 \Gamma \leq \frac{2^{4\locnummodes}}{4}\sqrt{|V_1||V_2|}
\end{equation}
and we have, due to the structure of $\rho$, that $\|\rho^{i,j}-\sigma_\rho^{i,j}\|_1 = \Gamma/|V_1||V_2|$.
\end{proof}

Before discussing the physical implications and the various relations to results that can be obtained for spin systems, let us end this section with the following comment.

\begin{com}[Optimality of system size scaling] 
 The system size scaling in Theorem~\ref{thm:monogamy} and \ref{thm:suppression_complete_bipartite} is optimal.
 To see this, consider the case $\locnummodes = 1$ and the state
 \begin{align}
  \rho = \frac{1}{2^{\syssize}}\left(\frac{1}{\sqrt{|V_1|}}\frac{1}{\sqrt{|V_2|}}\sum\limits_{j\in V_1}\sum\limits_{k\in V_2} i\maj_j^1\maj_{k}^2+\id\right),
\end{align}
for any disjoint $V_1,V_2\subset [\syssize]$ with $V_1\cup V_2 = [\syssize]$. Note that $\rho$ is indeed a state as by Lemma~\ref{lm:norm_of_linear_ferm_combinations}, we have that $\|\sum_{j\in A}\maj_j^\alpha\|=\sqrt{|A|}$ for any $A\subset [\syssize]$ (such that $\rho$ is positive) and it is Hermitian and of unit trace by construction. We, however, find for $j\in V_1$ and $k\in V_2$ that 
\begin{align}
 \rho^{j,k} = \frac{i}{4\sqrt{|V_1||V_2|}}\maj_j^1\maj_k^2+\frac{1}{4}\id, \quad \sigma_{\rho}^{j,k} = \frac{1}{4}\id
\end{align}
which directly yields that 
\begin{align}
 \|\rho^{i,j} -\sigma_{\rho}^{i,j}\|_1 = \frac{1}{\sqrt{|V_1||V_2|}},
\end{align}
saturating the system size scaling of the bound of Theorem~\ref{thm:suppression_complete_bipartite}. In order to saturate the scaling given in Theorem~\ref{thm:monogamy}, one needs to choose $V_1=\{1\}$ which gives rise to
\begin{align}
 \frac{1}{\syssize-1}\sum\limits_{j=2}^\syssize\|\rho^{1,j} -\sigma_{\rho}^{1,j}\|_1 = \frac{1}{\sqrt{\syssize-1}}.
\end{align}
\end{com} \label{com:opt}

\section{Monogamy of quantum correlations for distinguishable particles}
\label{sec:Monogamy of correlations for distinguishable particles}
%
Theorem~\ref{thm:monogamy} allows us to bound the anti-symmetric correlation a fermionic state can distribute over multiple subsystems. We now turn to study the correlations in systems of distinguishable particles for two reasons. First, we want to compare the strength of the suppression of correlations between fermionic systems and systems of distinguishable particles. Second, after erasing the antisymmetric correlations between different subsystems the remaining state is essentially a state of distinguishable particles. In order to connect this state to a product state over the subsystems, we need to bound the pontential remaining correlations.

In Ref.~\cite{Brandao2016}, 
the correlation structure of systems of finite dimensional distinguishable particles has been studied. As we explain in the following, the presented proofs are partially erroneous while the results are up to appropriate changes correct and the proof strategy developed in Ref.~\cite{Brandao2016} versatile and powerful 
enough in order to easily study settings beyond the ones considered in Ref.~\cite{Brandao2016}. As some of the proofs are significantly longer then ones discussed in Section~\ref{sec:Monogamy of fermionic two site correlation} we collect them below in Section~\ref{sec:dis_proofs} and first discuss the results here.
The following statement holds, as we show in Section~\ref{sec:dis_proofs} based on the arguments presented in Ref.\ \cite{Brandao2016}. With globally separable, a fully separable state is meant with respect to the vertices.

\begin{thm}[General product state approximation of 
distinguishable particles]\label{thm:gen_distinguish_psa}
Let $G\in\R^{\syssize\times \syssize}$ be a symmetric matrix with positive entries and $\sum_{i,j\in V}G_{i,j} = 1$. Define $\pi_j = \sum_{i\in V} G_{i,j}$ and $A_{i,j} = G_{i,j}/\pi_j$. 
 Then, for any $\rho\in\mathcal{D}((\C^d)^{\otimes \syssize})$ there exists a globally separable state $\sigma$ such that 
 \begin{equation}
  \EE\limits_{(i,j)\sim G}\|\rho^{i,j}-\sigma^{i,j}\|_1 \leq 47\left(d^4 \ln(d)\tr(A^2)\|\pi\|_2^2\right)^{1/5}+2\|\pi\|_2^2.
 \end{equation} 
\end{thm}
As we have explained in Section~\ref{sec:dis_proofs}, 
we have here not only corrected 
the proof of Theorem~\ref{thm:gen_distinguish_psa}, but also the result itself such that the bound in Theorem~\ref{thm:gen_distinguish_psa} differs from the one stated in Ref.\ \cite{Brandao2016}. From the above result we can easily bound product state approximations for general graphs as follows.

\begin{cor}[Bound to product state approximations for general graphs]
For any $\rho\in\mathcal{D}((\C^d)^{\otimes \syssize})$, there exists a globally separable state $\sigma$ such that 
\begin{equation}
  \EE\limits_{(i,j)\sim E}\|\rho^{i,j}-\sigma^{i,j}\|_1 
  \leq 47\left(d^4 \ln(d)\sum_{(i,j)\in E}\frac{1}{\eis{i}\eis{j}}\sum_i \frac{\eis{i}^2}{4\es^2}\right)^{1/5}+\sum_i \frac{\eis{i}^2}{2\es^2} .
 \end{equation} 
\end{cor}
\begin{proof}
For a given a graph $(V,E)$ with $V=[\syssize]$ define the matrix $G\in\R^{\syssize\times \syssize}$ with entries $G_{i,j} = (\delta_{(i,j)\in E}+\delta_{(j,i)\in E})/(2\es)$,  where we assume for simplicity that if $(i,j)\in E \Rightarrow (j,i)\notin E$. We then have that $G$ is symmetric and $\sum_{i,j}G_{i,j} =1$ with $\pi_j=\eis{i}/(2\es)$ such that $\|\pi\|_2^2 = \sum_i \eis{i}^2/
(4\es^2)$ 
%
and $\tr(A^2)=\sum_{(i,j)\in E}1/(\eis{i}\eis{j})$.
\end{proof}

Note that in the special case of $(V,E)$ being a $c$ regular graph or a star one obtains easily (see Ref.\ \cite{Brandao2016} for the $c$ regular graph result and Section~\ref{sec:spin_star} for the star-graph and a generalization to a complete bipartite graph in Section~\ref{sec:Special case: complete bipartite graphs}) the following 
substantially stronger results.

\begin{thm}[Product state approximation for distinguishable particles on special graphs]\label{thm:special_distinguish_psa}
For $(V,E)$ with $V = [\syssize]$ being $c$-regular it follows 
that for any $\rho\in\mathcal{D}((\C^d)^{\otimes \syssize})$ there exists a globally separable state $\sigma$ such that 
 \begin{equation}
  \EE\limits_{(i,j)\sim E}\|\rho^{i,j}-\sigma^{i,j}\|_1 \leq 12 \left(\frac{d^2\ln(d)}{c}\right)^{\frac{1}{3}}.
 \end{equation} 
For $(V,E)$ being a star in which $\syssize-1$ leaves couple to one root it follows that for any $\rho\in\mathcal{D}((\C^d)^{\otimes \syssize})$ there exists a globally separable state $\sigma$ such that 
 \begin{equation}
  \EE\limits_{(i,j)\sim E}\|\rho^{i,j}-\sigma^{i,j}\|_1 \leq 22 \left(\frac{d^2\ln(d)}{\syssize-1}\right)^{\frac{1}{3}}.
 \end{equation} 
\end{thm}

The above bounds hold for general states $\rho$ without any further structure required. In more structured settings, e.g., 
if the state are extendable, stronger bounds are known 
as we will discuss in Section~\ref{sec:extendibility}.



\section{Product state approximation to fermionic ground states}

Together with the monogamy of fermionic correlations, the product state approximation to distinguishable particles allows us to bound the error imposed by a mode-mean field approximation the ground state energy density of specific fermionic models.
To show this, 
assume a Hamiltonian defined on a graph $(V,E)$ as
\begin{equation}
 H = \sum\limits_{(i,j)\in E} h_{i,j}
\end{equation}
with $h_{i,j}$ being only a operator of the modes on sites $i$ and $j$ and of bounded strength, i.e., the operator norm of the Hamiltonian terms is bounded from above by~$\|h_{i,j}\|\leq 1$.
Further, assume that for any state $\rho\in\mathcal{D}((\C^d)^{\otimes \syssize})$ with $d = 2^\locnummodes$, we know that there is a separable state $\sigma$ in the vicinity such that 
\begin{align}
 \sum\limits_{(i,j)\in E}\frac{1}{|E|}\|\rho^{i,j}-\sigma^{i,j}\|_1\leq \epsilon.
\end{align}
Along the lines of the argumentation in Ref.~\cite{Krumnow2017}, we obtain the following result.

\begin{thm}[Product approximations for fermions]
 Given the above Hamiltonian and bound for distinguishable particles and any vertex cover $V^\prime$ of $(V,E)$, we obtain that there is a product state $\sigma=\sigma^1\otimes\dots\otimes\sigma^\syssize\in\mathcal{D}(\fock{\locnummodes \syssize})$ with 
 \begin{equation}
  \frac{1}{|E|}\left(E_\text{GS}-\tr(H\sigma)\right) \leq \sum\limits_{i\in V^\prime}\frac{2^{4\locnummodes}\sqrt{|E_i|}}{|E|} + \epsilon.
 \end{equation}
\end{thm}
\begin{proof}
 To prove the result denote with $\rho$ any ground state of $H$. Apply then Theorem~\ref{thm:monogamy} in order to conclude that 
 \begin{align}
  \sum\limits_{(i,j)\in E}\|\rho^{i,j}-\Xi(\rho)^{i,j}\|_1 \leq \sum\limits_{i\in V^\prime}2^{4\locnummodes}\sqrt{|E_i|}.
 \end{align}
By virtue of the Jordan Wigner transformation, 
we can interpret $\Xi(\rho)$ as a state in $\mathcal{D}((\C^{d})^{\otimes \syssize})$ with $d=2^\locnummodes$ of distinguishable particles.
Let $\tau\in\mathcal{D}((\C^{d})^{\otimes \syssize})$ be a separable state approximating $\Xi(\rho)$ locally up to $\epsilon$, i.e.,
 \begin{align}
  \sum\limits_{(i,j)\in E}\|\Xi(\rho)^{i,j}-\tau^{i,j}\|_1 \leq |E|\epsilon.
 \end{align}
Define $\omega = \Xi(\tau)$ it follows from the contractiveness of the 1-norm under the application of channels that 
\begin{align}
  \sum\limits_{(i,j)\in E}\|\Xi(\rho)^{i,j}-\omega^{i,j}\|_1 \leq |E|\epsilon
 \end{align}
 as $\Xi(\Xi(\rho)) = \Xi(\rho)$. Furthermore, $\omega$ is due to the separability of $\tau$ a convex combination of mode product states of the form $\omega^1_i\otimes\dots \otimes\omega^\syssize_i\in\mathcal{D}(\fock{\locnummodes\syssize})$. Hence, we have by the triangle inequality
 \begin{align}
  \sum\limits_{(i,j)\in E}\|\rho^{i,j}-\omega^{i,j}\|_1 \leq \sum\limits_{i\in V^\prime}2^{4\locnummodes}\sqrt{|E_i|} +|E|\epsilon.
 \end{align}
 Pick now $\sigma = \omega_{m}^1\otimes\dots\otimes\omega_m^\syssize$ 
 with $m=\argmin_{i}\tr(H\omega^1_i\otimes\dots \otimes \omega^\syssize_i)$ which yields
 \begin{align}
  E_\text{GS} - \tr(H\sigma) = \tr(H\rho)-\tr(H\sigma) &\leq \tr(H\rho)-\tr(H\omega) \\
  &=\sum\limits_{(i,j)\in E} \tr(h_{i,j}[\rho^{i,j}-\omega^{i,j}]) \nonumber \\
  &\leq \sum\limits_{(i,j)\in E} \|h_{i,j}\|\|\rho^{i,j}-\omega^{i,j}\|_1\nonumber  \\
  &\leq \sum\limits_{i\in V^\prime}2^{4\locnummodes}\sqrt{|E_i|} +|E|\epsilon.
  \nonumber 
 \end{align}
\end{proof}
It is worth noting that these variational upper bounds are complemented by lower bounds of ground state energy densities
\begin{equation}
  \frac{1}{|E|}\left(E_\text{GS}-\tr(H\sigma)\right) \geq C_1
 \end{equation}
 that are again constant
for suitable constants 
$C_1$ that can be efficiently computed \cite{LowerBounds}
(the situation is 
 particularly clear for translationally invariant models). In this sense, the true ground state energy density is 
 ``sandwiched'' by the two respective bounds. The above findings have the following immediate implications that we specifically flesh out in detail as corollaries.

\begin{cor}[Product approximations 
for two-local fermionic systems]
\label{cor:energy_full_edge}
For a two local Hamiltonian $H$ with general interaction graph $(V,E)$ we hence obtain that there exists a mode product state $\sigma$ such that 
\begin{align}
 \frac{1}{|E|}\left(E_\text{GS} - \tr(H\sigma)\right) \leq \sum\limits_{i\in V^\prime}\frac{2^{4\locnummodes}\sqrt{|E_i|}}{|E|} + 47\left(2^{4\locnummodes} \locnummodes\sum_{(i,j)\in E}\frac{1}{|E_i||E_j|}\sum\limits_{i\in V}\frac{|E_i|^2}{|E|^2}\right)^{1/5}+2\sum\limits_{i\in V} \frac{|E_i|^2}{|E|^2}
\end{align}
with $V^\prime$ being any vertex cover of $(V,E)$.
\end{cor}

Let us denote for a specific graph $(V,E)$ and Hamiltonian $H$ the error of the ground state energy density when employing a mode product state approximation by 
\begin{equation}
\delta=(E_\text{GS}-\min_{\sigma\in\mathcal{D}_{\rm prod}(\fock{\locnummodes\syssize})}\tr(H\sigma))/|E|,
\end{equation}
where ${D}_{\rm prod}(\fock{\locnummodes\syssize})$ is the set of mode product states in ${D}(\fock{\locnummodes\syssize})$.
For specific lattices and models we obtain the following bounds.

\begin{cor}[Ground state energy density]
  On a $c$-regular lattice the  ground state energy density can be approximated by a mode mean field 
  \begin{align}
   \delta_\text{\rm $c$-regular} = \frac{ 2^{4\locnummodes}}{8\sqrt{c} } + 12 \left(\frac{2^{2\locnummodes}\locnummodes}{c}\right)^{\frac{1}{3}}, 
  \end{align}
  in particular, for the $\spdim$ dimensional spinless Fermi Hubbard model on a square lattice
  \begin{equation}
 H = \sum\limits_{i=1}^\syssize\sum\limits_{k=1}^\spdim \left[ t\left(\fer^\dag_{i}\fer_{i+\hat{e}_k} + \fer^\dag_{i}\fer_{i-\hat{e}_k}\right) + U \left(n_{i}n_{i+\hat{e}_k}+n_{i}n_{i-\hat{e}_k}\right)\right],
\end{equation}
  where $i\pm \hat{e}_k$ denotes the next nearest neighbouring site to $i$ in $\pm \hat{e}_k$ direction, we obtain 
  \begin{align}
   \delta_\text{\rm Spinless Fermi-Hubbard} = \max(t,U) \left[\frac{ 2^{4\locnummodes}}{8\sqrt{2\spdim} } + 12 \left(\frac{2^{2\locnummodes}\locnummodes}{2\spdim}\right)^{\frac{1}{3}} \right] .
  \end{align}
  Furthermore, on a star-like graph, we can bound
  \begin{align}
   \delta_\text{\rm Star} = \frac{ 2^{4\locnummodes}}{4\sqrt{\syssize-1} } + 18 \left(\frac{2^{2\locnummodes}\locnummodes}{\syssize-1}\right)^{\frac{1}{3}}.
  \end{align}
\end{cor}

It is, however, important to stress that many fermionic models studied, for instance, in condensed matter physics are not of this form and might even in the infinite dimensional case not allow for an faithful mean field treatment (at least according to the above theorems). Consider, for instance, the spinful Fermi Hubbard model in $\spdim$ dimensions
\begin{equation}
 H = \sum\limits_{i=1}^\syssize \left[\sum\limits_{s=1,2}\sum\limits_{k=1}^\spdim t \left(\fer^\dag_{i,s}\fer_{i+\hat{e}_k,s} + \fer^\dag_{i,s}\fer_{i-\hat{e}_k,s}\right) + U n_{i,1}n_{i,2}\right].
\end{equation}
Here, due to the on-site character of the interactions, the typical energy to be expected scales as $\syssize$ and not with $\syssize \spdim$ as above such that $\lim_{\spdim\rightarrow\infty}E_{\text{GS}}/|E| \rightarrow 0$ rendering the approximation result trivial in this case.

Next to models from condensed matter physics, it was shown 
in
Ref.~\cite{Babbush2018} that the general Hamiltonian of 
electron structure theory of molecules used in quantum 
chemistry can be brought into the form
\begin{align}
 H = \sum\limits_{i,j=1}^\syssize\sum\limits_{\alpha,\beta=1}^{\locnummodes}\left[t_{i,\alpha,j,\beta}\fer_{i,\alpha}^\dag\fer_{j,\beta}+v_{i,\alpha,j,\beta}n_{i,\alpha}n_{j,\beta}\right]\label{eq:H_QC}
\end{align}
when discretized, for instance, using a dual-Fourier basis. As argued above, in the presence of a specific structures in Eq.\ \eqref{eq:H_QC}, e.g.,
regularity and suitable coefficients, Corollary~\ref{cor:energy_full_edge} predicts that the 
ground state can be approximated via a mode-product state. Note, however, that for the general case, the normalization yet again has to be carefully considered. The ground state energy of \eqref{eq:H_QC} will typically scale with the number particles $K$ and usually $\syssize\propto K$. Hence, the bounding the ground state energy with an error of order $|E|$ that generically scales as $\syssize^2$ quickly yields only trivial bound.

Note that these limitations can be understood from a more general point of view. By Theorem~\ref{thm:monogamy}, mode product states will be able to faithfully approximate the ground state of models for which terms which are sensitive to the anti-symmetric character of the fermionic wavefunction are suppressed by an argument similar to the one in Theorem~\ref{thm:monogamy}. Hence, models which show anti-symmetric correlations, balance anti-symmetry-sensitive terms such as hopping against density-density-like contributions appropriately.
On such models, the fermionic character needs to be taken truly into account in a mean field method as for instance by a general Hartree-Fock approach or when using variational approaches as with set of states provided by Ref.\ \cite{Kraus2013} for permutation invariant systems.


\section{Symmetric extendibility: Differences between fermions and distinguishable particles}
\label{sec:extendibility}

The monogamy of entanglement can also be captured by the so-called {\it symmetric extendibility\/} (or {\it shareability\/}), that describes the degree to which a bipartite state can be shared between multiple parties \cite{terhal2003symmetric}. In case of distinguishable particles, a bipartite quantum state $\omega \in\mathcal{D}(\C^d \otimes \C^d)$  is said to be $(n,k)$-extendible if for a complete bipartite graph
 $(V,E)$ with components $V_1$ and $V_2$  
 of size  $|V_1|=n$ and $|V_2|=k$,  
 if there exists a global quantum state $\rho \in\mathcal{D} \left((\C^d)^{\otimes n}\otimes(\C^d)^{\otimes k}\right)$ 
 such that $\rho^{i,j} = \omega $ for all $i \in V_1$ and all $j \in V_2 $.
 Without loss of generality one can assume that the state $\rho$ is invariant with respect to permuting the Hilbert space factors in the tensor products $\mathcal{H}^{\otimes n}$ corresponding to $V_1 $ (and also the factors in the tensor products $\mathcal{H}^{\otimes k}$ corresponding to $V_2 $), hence the name {\it symmetric\/} extendibility. 


A state's maximal extendibility numbers serve as an upper bound of its distance to the set of separable states. In particular, for the one-sided symmetric extendibility,  the following  bound holds \cite{brandao2011faithful}:  Let $\omega \in\mathcal{D}(\mathcal{H}_A\otimes\mathcal{H}_B)$ be a state that is $(1,k)$-extendible, then there exists a separable state $\sigma \in\mathcal{D}(\mathcal{H}_A\otimes\mathcal{H}_B)$ such that 
\begin{equation} \label{eq:sym1}
    \| \omega-\sigma \|_{1} \leq \frac{4 d^2}{k}.
\end{equation}
 This bound is close to being tight, as there are $(1,k)$-extendible states that are  $O(d/k)$ away from the set of separable states \cite{Christandl2007}. As for two-sided extendibility, it was shown that there are $(n,k)$-extendible states that are in one-norm only  $O(d/{\rm max}\{n,k\})$ away from the separable set \cite{jakab2022extendibility}, thus one cannot obtain a better general bound than 
\begin{equation} \label{eq:sym2}
    \| \rho-\sigma \|_{1} \leq \frac{C}{{\rm max}\{n,k\}},
\end{equation}
where $C>0$ is a dimension dependent constant that may be smaller than the  analogous constant for one-sided extendibility (and, in particular, $C \le 4 d^2$).

One can also naturally define extendibility 
for fermionic systems. Let $(V,E)$ be a complete bipartite graph with components with components $V_1$ and $V_2$   of size  $|V_1|=n$ and $|V_2|=k$.
A bipartite state $\omega \in \mathcal{D}(\fock{2\locnummodes})$ is said to be $(n,k)$-extendable if there exists a state $\rho\in\mathcal{D}(\fock{(n+k)\locnummodes})$ with the restrictions $\rho^{i,j} = \omega$ for all $i \in V_1$ and all $j \in V_2 $. 
Noting that only a totally even states can be separable and also that totally even states can be treated as states of distinguishable particles, we can use Eqs. \eqref{eq:sym1}, \eqref{eq:sym2} and Theorem \ref{thm:suppression_complete_bipartite}
to arrive to the following corollaries.  

\begin{cor}[One-sided symmetric extendibility for fermions]\label{cor:symm_extendibility}
Let a bipartite state $\omega \in \mathcal{D}(\fock{2\locnummodes})$ be $(1,k)$-extendable, then there exists a bipartite separable state $\sigma \in\mathcal{D}(\fock{2\locnummodes})$ such that 
\begin{equation} \label{eq:sym3}
    \| \omega-\sigma \|_{1} \leq \frac{2^{4\locnummodes}}{4}\frac{1}{\sqrt{k}} + \frac{8}{k}.
\end{equation}
 \end{cor}
\begin{cor}[Two-sided symmetric extendibility for fermions]\label{cor:symm_extendibility}
Let a bipartite state $\omega \in \mathcal{D}(\fock{2\locnummodes})$ be $(n,k)$-extendable, then there exists a bipartite separable state $\sigma \in\mathcal{D}(\fock{2\locnummodes})$ such that 
\begin{equation} \label{eq:sym4}
    \| \omega-\sigma \|_{1} \leq \frac{2^{4\locnummodes}}{4}\frac{1}{\sqrt{nk}} + \frac{8}{k}.
\end{equation}
In particular, if  $n =k$
\begin{equation} \label{eq:sym5}
    \| \omega-\sigma \|_{1} \leq \frac{2^{4\locnummodes} +32}{4}\frac{1}{k}.
\end{equation}
 \end{cor}

The example in Comment \ref{com:opt} shows that the bounds given above are close to optimal.  Contrasting Eqs. \eqref{eq:sym1}, \eqref{eq:sym2} with Eqs. \eqref{eq:sym3}, \eqref{eq:sym4}, we can note that in case of fermions one-sided symmetric extendibility does not imply as stringent conditions on closeness to separability as for distinguishable particles. On the other hand, two-sided  extendability gives similarly stringent bounds.

\section{Outlook}

In this work, we have presented results on the extendibility of fermionic quantum states. We have seen that if correlation in a fermionic quantum state is shared over too many subsystems, e.g., lattice sites, this necessarily means that these correlations cannot be particularly strong. 
Our result can be seen as an immediate analog of the results on spin models of Ref.\ \cite{Brandao2016} on 
product-state approximations to quantum states of distinguishable particles. The extension to fermionic quantum states relies on a suppression of anti-symmetric correlations shown here and the results of Ref.\ \cite{Brandao2016}. It implies certificates on the approximability of the ground state energy density using of suitable models using product states approximations. 
Furthermore, we show and discuss that the suppression of correlations for extendable states differ for fermionic and distinguishable particle quantum states, leaving more room for correlations in the fermionic setting. This work hence makes the point that if fermionic systems have many neighbors they interact with, a product state will provide a good approximation of the ground state. This is expected to be particularly relevant for problems in quantum chemistry where such systems are common.

This work connects to a number of recent results on interacting fermionic quantum systems. Prominently, a fermionic variant of the NLTS conjecture \cite{NLTS} 
has been settled \cite{FermionicNLTS, PhysRevA.109.052431}. This insight settles
the question whether fermionic Hamiltonians  with no low-energy trivial states
exist. In particular, a Hamiltionian on $N$ modes is presented 
such that for an arbitrary $3N/2$-
fermion state $\rho$ such that ${\rm Tr}(\rho H)< \varepsilon N$, 
using arbitrary Gaussian initialization and an arbitrary number of auxiliary modes, 
any fermionic circuit
that prepares $\rho$ has depth $\Omega(\log(N))$.
This is relevant for the present work, as for a 
Hamiltonian with no low-energy trivial states, a fermionic product state will not provide a good ground state approximation. The subtle differences
between systems of distinguishable spin systems and fermionic systems that are 
explored in the technical part of this work are also reflected by other 
features that have been studied recently. Prominently,
Ref.\ \cite{AnschuetzNonGlassy} suggests that low-temperature strongly interacting fermions, unlike spins, belong in a classically nontrivial yet quantumly easy phase, as judged by the algorithms of Ref.\ \cite{HastingsFermions}. That work also suggests that fermions would be simpler to see quantum advantages for, as quantum algorithms perform better for classically hard problems.
Once again, this gives further significance
to the insight that fermionic and spin systems  differ in their commutation index. At the end of the day, many of those differences can be attributed to non-local strings when expressing fermionic as spin systems. 
Also, the algorithm presented for approximating low-energy
states in Ref.\ \cite{HastingsFermions} is not applicable for
spins in the same way. In a yet similar fashion, bosonic and fermionic 
central limit theorems \cite{AnalyticalQuench,Marek} are distinctly different, in ways that can be traced back to the fermionic anti-commutation relations, statements that also explain the behavior of 
out-of-equilibrium dynamics. 

In addition, a deeper investigation on how the obtained bounds relate to systems of quantum chemistry is needed. Here either the investigation of concrete interaction-graph structures or the derivation of more general results, e.g., the case of a lower bound coordination number or known average and variance of the coordination number of the graph would be of interest and can be developed along the lines of the presented results. It is the hope that the present work 
invites further endeavours of studying the differences and 
similarities between fermionic and spin quantum  systems.

\section{Acknowledgements}
We acknowledge discussions with Jonas Helsen, David Divincenzo, and Eric Anschuetz. This works has been supported by the DFG (CRC 183), the
Quantum Flagship (PasQuans2, OpenSuperQPlus), the Munich Quantum Valley,
the BMBF (QSolid, FermiQP), and the ERC (DebuQC). For the QuantERA project HQCC, this is the result of a joint node collaboration.

\bibliographystyle{quantum}


\begin{thebibliography}{10}

\bibitem{CramerQuantumChemistry}
C.~J. Cramer.
\newblock ``Essentials of computational chemistry: Theories and models''.
\newblock \href{https://dx.doi.org/10.1007/s00214-002-0380-8}{John Wiley and
  Sons}. Chichester~(2002).

\bibitem{Lilienfeld}
G.~Montavon, M.~Rupp, V.~Gobre, A.~Vazquez-Mayagoitia, K.~Hansen,
  A.~Tkatchenko, K.-R. M{\"u}ller, and O.~Anatole von Lilienfeld.
\newblock ``Machine learning of molecular electronic properties in chemical
  compound space''.
\newblock \href{https://dx.doi.org/10.1088/1367-2630/15/9/095003}{New J. Phys.
  {\bf 15}, 095003}~(2013).

\bibitem{QuantumMaterials}
E.~Kaxiras and J.~D. Joannopoulosy.
\newblock ``Quantum theory of materials''.
\newblock \href{https://dx.doi.org/10.1017/9781139030809}{Cambridge University
  Press}. Cambridge~(2019).

\bibitem{draxl_scheffler_2018}
C.~Draxl and M.~Scheffler.
\newblock ``Nomad: The fair concept for big data-driven materials science''.
\newblock \href{https://dx.doi.org/10.1557/mrs.2018.208}{MRS Bulletin {\bf 43},
  676–682}~(2018).

\bibitem{Orus-AnnPhys-2014}
R.~Or\'us.
\newblock ``A practical introduction to tensor networks: Matrix product states
  and projected entangled pair states''.
\newblock \href{https://dx.doi.org/10.1016/j.aop.2014.06.013}{Ann. Phys. {\bf
  349}, 117--158}~(2014).

\bibitem{ReiherChemistry}
A.~Baiardi and M.~Reiher.
\newblock ``The density matrix renormalization group in chemistry and molecular
  physics: Recent developments and new challenges''.
\newblock \href{https://dx.doi.org/10.1063/1.5129672}{J. Chem. Phys. {\bf 152},
  040903}~(2020).

\bibitem{TroyerChemistry}
S.~Keller, M.~Dolfi, M.~Troyer, and M.~Reiher.
\newblock ``{An efficient matrix product operator representation of the quantum
  chemical Hamiltonian}''.
\newblock \href{https://dx.doi.org/10.1063/1.4939000}{J. Chem. Phys. {\bf 143},
  244118}~(2015).

\bibitem{Krumnow2016}
C.~Krumnow, L.~Veis, \"O. Legeza, and J.~Eisert.
\newblock ``Fermionic orbital optimisation in tensor network states''.
\newblock \href{https://dx.doi.org/10.1103/PhysRevLett.117.210402}{Phys. Rev.
  Lett. {\bf 117}, 210402}~(2016).

\bibitem{Pfeffer}
S.~Szalay, M.~Pfeffer, V.~Murg, G.~Barcza, F.~Verstraete, R.~Schneider, and Ö.
  Legeza.
\newblock ``Tensor product methods and entanglement optimization for ab initio
  quantum chemistry''.
\newblock \href{https://dx.doi.org/10.1002/qua.24898}{Int. J. Quant. Chem. {\bf
  115}, 1342}~(2015).

\bibitem{AreaReview}
J.~Eisert, M.~Cramer, and M.~B. Plenio.
\newblock ``Area laws for the entanglement entropy''.
\newblock \href{https://dx.doi.org/10.1103/RevModPhys.82.277}{Rev. Mod. Phys.
  {\bf 82}, 277}~(2010).

\bibitem{McClean_2016}
J.~R McClean, J.~Romero, R.~Babbush, and A.~Aspuru-Guzik.
\newblock ``The theory of variational hybrid quantum-classical algorithms''.
\newblock \href{https://dx.doi.org/10.1088/1367-2630/18/2/023023}{New J. Phys.
  {\bf 18}, 023023}~(2016).

\bibitem{Mazziotti2002}
D.~A. Mazziotti.
\newblock ``Variational minimization of atomic and molecular ground-state
  energies via the two-particle reduced density matrix''.
\newblock \href{https://dx.doi.org/10.1103/PhysRevA.65.062511}{Phys. Rev. A
  {\bf 65}, 062511}~(2002).

\bibitem{PhysRevLett.59.121}
W.~Metzner and D.~Vollhardt.
\newblock ``{Ground-state properties of correlated fermions: Exact analytic
  results for the Gutzwiller wave function}''.
\newblock \href{https://dx.doi.org/10.1103/PhysRevLett.59.121}{Phys. Rev. Lett.
  {\bf 59}, 121--124}~(1987).

\bibitem{Christandl2007}
M.~Christandl, R.~K{\"o}nig, G.~Mitchison, and R.~Renner.
\newblock ``One-and-a-half quantum de {F}inetti theorems''.
\newblock \href{https://dx.doi.org/10.1007/s00220-007-0189-3}{Commun. Math.
  Phys. {\bf 273}, 473}~(2007).

\bibitem{Krumnow2017}
C.~Krumnow, Z.~Zimbor\'as, and J.~Eisert.
\newblock ``{A fermionic de Finetti theorem}''.
\newblock \href{https://dx.doi.org/10.1063/1.4998944}{J. Math. Phys. {\bf 58},
  122204}~(2017).

\bibitem{Monogamy}
B.~M. Terhal.
\newblock ``Is entanglement monogamous?''.
\newblock \href{https://dx.doi.org/10.1147/rd.481.0071}{IBM J. Res. Dev. {\bf
  48}, 71--78}~(2004).

\bibitem{Brandao2016}
F.~G. S.~L. Brandao and A.~W. Harrow.
\newblock ``Product-state approximations to quantum states''.
\newblock \href{https://dx.doi.org/10.1007/s00220-016-2575-1}{Commun. Math.
  Phys. {\bf 342}, 47}~(2016).

\bibitem{AnschuetzNonGlassy}
E.~R. Anschuetz, C.-F. Chen, B.~T. Kiani, and R.~King.
\newblock ``Strongly interacting fermions are non-trivial yet
  non-glassy''~(2024).
\newblock  \href{http://arxiv.org/abs/2408.15699}{arXiv:2408.15699}.

\bibitem{HastingsFermions}
M.~B. Hastings and R.~O'Donnell.
\newblock ``{Optimizing strongly interacting fermionic Hamiltonians}''~(2021).
\newblock  \href{http://arxiv.org/abs/2110.10701}{arXiv:2110.10701}.

\bibitem{NLTS}
M.~H. Freedman and M.~B. Hastings.
\newblock ``Quantum systems on non-$k$-hyperfinite complexes: A generalization
  of classical statistical mechanics on expander graphs''.
\newblock \href{https://dx.doi.org/10.26421/QIC14.1-2-9}{Quantum Inf. Comp.
  {\bf 14}, 144}~(2014).

\bibitem{PhysRevA.109.052431}
Y.~Herasymenko, A.~Anshu, B.~M. Terhal, and J.~Helsen.
\newblock ``{Fermionic Hamiltonians without trivial low-energy states}''.
\newblock \href{https://dx.doi.org/10.1103/PhysRevA.109.052431}{Phys. Rev. A
  {\bf 109}, 052431}~(2024).

\bibitem{FermionicNLTS}
A.~Anshu, N.~P. Breuckmann, and C.~Nirkhe.
\newblock ``{NLTS Hamiltonians from good quantum codes}''~(2022).
\newblock  \href{http://arxiv.org/abs/2206.13228}{arXiv:2206.13228}.

\bibitem{LowerBounds}
J.~Eisert.
\newblock ``A note on lower bounds to variational problems with
  guarantees''~(2023).
\newblock  \href{http://arxiv.org/abs/2301.06142}{arXiv:2301.06142}.

\bibitem{Babbush2018}
R.~Babbush, N.~Wiebe, J.~McClean, J.~McClain, H.~Neven, and G.~K.-L. Chan.
\newblock ``Low-depth quantum simulation of materials''.
\newblock \href{https://dx.doi.org/10.1103/PhysRevX.8.011044}{Phys. Rev. X {\bf
  8}, 011044}~(2018).

\bibitem{Kraus2013}
C.~V. Kraus, M.~Lewenstein, and J.~I. Cirac.
\newblock ``{Ground states of fermionic lattice Hamiltonians with permutation
  symmetry}''.
\newblock \href{https://dx.doi.org/10.1103/PhysRevA.88.022335}{Phys. Rev. A
  {\bf 88}, 022335}~(2013).

\bibitem{terhal2003symmetric}
B.~M. Terhal, A.~C. Doherty, and D.~Schwab.
\newblock ``Symmetric extensions of quantum states and local hidden variable
  theories''.
\newblock \href{https://dx.doi.org/10.1103/PhysRevLett.90.157903}{Phys. Rev.
  Lett. {\bf 90}, 157903}~(2003).

\bibitem{brandao2011faithful}
F.~G. S.~L. Brand\~{a}o, M.~Christandl, and J.~Yard.
\newblock ``Faithful squashed entanglement''.
\newblock \href{https://dx.doi.org/10.1007/s00220-011-1302-1}{Commun. Math.
  Phys. {\bf 306}, 805--830}~(2011).

\bibitem{jakab2022extendibility}
D.~Jakab, A.~Solymos, and Z.~Zimborás.
\newblock ``{Extendibility of Werner states}''~(2022).
\newblock  \href{http://arxiv.org/abs/2208.13743}{arXiv:2208.13743}.

\bibitem{AnalyticalQuench}
M.~Cramer, C.~M. Dawson, J.~Eisert, and T.~J. Osborne.
\newblock ``Exact relaxation in a class of non-equilibrium quantum lattice
  systems''.
\newblock \href{https://dx.doi.org/10.1103/PhysRevLett.100.030602}{Phys. Rev.
  Lett. {\bf 100}, 030602}~(2008).

\bibitem{Marek}
M.~Gluza, C.~Krumnow, M.~Friesdorf, C.~Gogolin, and J.~Eisert.
\newblock ``Gaussification and equilibration in free {Hamiltonian} systems''.
\newblock \href{https://dx.doi.org/10.1103/PhysRevLett.117.190602}{Phys. Rev.
  Lett. {\bf 117}, 190602}~(2016).

\bibitem{Raghavendra2012}
P.~Raghavendra and N.~Tan.
\newblock ``{Approximating {CSP}s with global cardinality constraints using
  {SDP} hierarchies}''.
\newblock In Proceedings of the Twenty-third Annual {ACM-SIAM} Symposium on
  Discrete Algorithms.
\newblock \href{https://dx.doi.org/10.48550/arXiv.1110.1064}{Page 373}.
\newblock SODA '12Philadelphia, PA, USA~(2012). Society for Industrial and
  Applied Mathematics.

\bibitem{PaulPrivate}
P.~Boes.
\newblock private communication~(2018).

\bibitem{Kochar2001}
S.~C. Kochar and R.~Korwar.
\newblock ``On random sampling without replacement from a finite population''.
\newblock \href{https://dx.doi.org/10.1023/A:1014693702392}{Ann. Inst. Statist.
  Math. {\bf 53}, 631}~(2003).

\end{thebibliography}

\newpage
\appendix
\section{Connection to the results of Brandao and Harrow}
\label{sec:difference_of_results}
In this collection of appendices, we present mathematical details of the arguments laid out in the main text. In this section, we discuss the differences of our results from the proof and results in 
Ref.~\cite{Brandao2016}. We adapt, therefore, 
the notation of Ref.\ \cite{Brandao2016} here.
Theorem 9 in Ref.~\cite{Brandao2016}
states the following.

\begin{thm}[Product state approximation by Brandao and Harrow]
 Let $G$ be a symmetric matrix with positive entries and $\sum_{i,j\in V}G_{i,j} = 1$. Define $\pi_j = \sum_{i\in V} G_{i,j}$ and $A_{i,j} = G_{i,j}/\pi_j$. 
 Then, for any $\rho\in\mathcal{D}((\C^d)^{\otimes n})$ there exists a globally separable state $\sigma$ such that
 \begin{equation}
  \EE\limits_{(i,j)\sim G}\|\rho^{Q_i,Q_j}-\sigma^{Q_i,Q_j}\|_1 \leq 14 \left(d^4\ln(d)\tr[A^2]\|\pi\|_2^2\right)^{1/8} + \|\pi\|_2^2.
 \end{equation} 
\end{thm}
In the main text, we have argued that the theorem is (up to minor changes) correct, however, claimed that the in Ref.~\cite{Brandao2016} presented proof, yet, is erroneous. 

The for our purposes essential flaw in the proof in Ref.~\cite{Brandao2016} appears when they define for a given probability distribution $\mu$ on the universe $[n]$ a distribution over distinct $m$-tuples of elements in $[n]$ as
\begin{equation}
 \mu^{\wedge m}(i_1,\dots,i_m) = \begin{cases}
                           0 &\text{if }i_1,\dots,i_m\text{ are not all distinct,} \\
                           \frac{\mu(i_1)\dots\mu(i_m)}{\sum_{j_1,\dots,j_m\text{ distinct}}\mu(j_1)\dots\mu(j_m)} &\text{otherwise}
                          \end{cases}.\label{defWrongSWOR}
\end{equation}
The distribution $\mu^{\wedge m}$ provides a valid distribution over distinct tuples, however it is not as claimed in Ref.~\cite{Brandao2016} the distribution of drawing the $m$-tuple $(i_1,\dots,i_m)$ without replacement according to $\mu$. The distribution of this process is given by
\begin{equation}
 \mu^{\wor m}(i_1,\dots,i_m) = \begin{cases}
                           0 &\text{if }i_1,\dots,i_m\text{ are not all distinct,} \\
                           \frac{\mu(i_1)\dots\mu(i_m)}{(1-\mu(i_1))\dots(1-\mu_{i_1}-\dots-\mu(i_{m-1}))} &\text{otherwise.}
                          \end{cases}\label{defCorrSWOR}
\end{equation}
The distinction between the two distribution in the context of the proof in Ref.~\cite{Brandao2016} is important, as it exploits the concrete shape of the used distribution and relies on the relation
\begin{equation}
 \sum\limits_{i_{m+1}\neq i_1,\dots,i_m}\mu^{\wor m+1}(i_1,\dots,i_m,i_{m+1}) = \mu^{\wor m}(i_1,\dots,i_m)
\end{equation}
which holds for sampling an ordered tuple without replacement, but fails to hold for $\mu^{\wedge m}$ and renders the proof in Ref.~\cite{Brandao2016} inaccurate. 
In Section~\ref{sec:dis_proofs}, we show that it is possible to follow the proof of Ref.~\cite{Brandao2016}, replace \eqref{defWrongSWOR} with \eqref{defCorrSWOR} and have indicated how to overcome any appearing discrepancies.

\section{Proof for the suppression of correlations for distinguishable particles}\label{sec:dis_proofs}

In the following we prove the different theorems of Section~\ref{sec:Monogamy of correlations for distinguishable particles}. The proofs are strongly based on the techniques and insights obtained in Ref.~\cite{Brandao2016} and we adapt much of their notation here. 
In Section~\ref{sec:difference_of_results} we discuss in what sense our result and proof deviates from the result in Ref.~\cite{Brandao2016}.
%
Given a set of classical random variables $X_1,\dots,X_n$ with a joint distribution $p$. Let $S$ denote the von Neumann entropy then we define the mutual information between subsets of the random variables as
\begin{align}
 I(X_{i_1}\dots X_{i_k}:X_{j_1}\dots X_{j_l})_p := S(X_{i_1}\dots X_{i_k}) + S(X_{j_1}\dots X_{j_l}) - S(X_{i_1}\dots X_{i_k}X_{j_1}\dots X_{j_l}).
\end{align}
Furthermore for $C\subset$ we denote with $p^{X_C}$ the marginal distribution of all $X_i$ with $i\in C$. We then define the conditional mutual information as 
\begin{align}
 I(a:b|C) = I(X_a:X_b|X_C):= \EE\limits_{x\sim p^{X_C}} I(X_a:X_b)_{p_x}
\end{align}
where $p_x$ denotes $p$ conditioned on $X_C=x$.
Below we exploit the chain rule for mutual information which states
\begin{align}
 I(X_a:X_bX_c) = I(X_a:X_c)+ I(X_a:X_b|X_c).
\end{align}

First we show a variant of the decoupling lemma of Refs.~\cite{Raghavendra2012,Brandao2016}.

\begin{lm}[Variant of the decoupling lemma]\label{lm:decoupling_star}
Let $n,k\in\N$ with $n>k\geq 1$. Assume a probability distribution $\pi,\mu$ over $[n]$ with $k<|\supp(\mu)|$ and denote by $\mu^{\wor k}$ the distribution of drawing $k$-tuples with entries in $[n]$ without replacement. Furthermore, let $X_1,\dots,X_n$ be random variables with $d$ outcomes each with a joint distribution $p$. Then, we find
\begin{equation}
 \EE\limits_{0\leq k^\prime < k}\EE\limits_{(c_1,\dots,c_{k^\prime})\sim\mu^{\wor k^\prime}}\sum\limits_{a, b}\pi(a)\mu(b) I(X_a:X_b|X_{c_1}\dots X_{c_{k^\prime}})\leq \frac{1}{k}\EE\limits_{i\sim\pi}I(X_i:X_{-i})
\end{equation}
with $X_{-i}$ denoting the collection of all random variables except $X_i$
and that there exists a $k^\prime\leq k$, $C=(c_1,\dots,c_{k^\prime})\subset [n]$ such that 
\begin{equation}
 \sum\limits_{a\neq b}\pi(a)\mu(b) I(X_a:X_b|X_{c_1}\dots X_{c_{k^\prime}})\leq \frac{\ln(d)}{k}.
\end{equation}

\end{lm}

\begin{proof}
Using the monotonicity under partial traces we find
 \begin{align}
 \frac{1}{k}\EE\limits_{i\sim\pi}\EE\limits_{\substack{(j_1,\dots,j_k)\sim\mu^{\wor k}\\j_l\neq i}} I(X_i:X_{j_1}\dots X_{j_k}) &\leq \frac{1}{k}\EE\limits_{i\sim\pi}\EE\limits_{\substack{(j_1,\dots,j_k)\sim\mu^{\wor k}\\j_l\neq i}} I(X_i:X_{-i})\\
 \nonumber
 &\leq\frac{1}{k}\EE\limits_{i\sim\pi} I(X_i:X_{-i}) &\text{(by \eqref{eqMarginalProp})} \\
 \nonumber
 &\leq\frac{\ln(d)}{k}.
\end{align}
Then, by repetitive use of the chain rule of mutual information,  we get
\begin{align}
 \frac{1}{k}\EE\limits_{i\sim\pi}\EE\limits_{\substack{(j_1,\dots,j_k)\sim\mu^{\wor k}\\j_l\neq i}}& I(X_i:X_{j_1}\dots X_{j_k})\nonumber\\
 &= \frac{1}{k}\EE\limits_{i\sim\pi}\EE\limits_{\substack{(j_1,\dots,j_k)\sim\mu^{\wor k}\\j_l\neq i}}\bigl[I(X_i:X_{j_k}|X_{j_1}\dots X_{j_{k-1}})+I(X_i:X_{j_1}\dots X_{j_{k-1}})\bigr]\\
 \nonumber
 &= \frac{1}{k}\EE\limits_{i\sim\pi}\EE\limits_{\substack{(j_1,\dots,j_k)\sim\mu^{\wor k}\\j_l\neq i}}\left[\sum\limits_{k^\prime=0}^{k-1} I(X_i:X_{j_{k^\prime+1}}|X_{j_1}\dots X_{j_{k^\prime}})\right]\\
  \nonumber
 &=\EE\limits_{i\sim\pi}\EE\limits_{\substack{(j_1,\dots,j_k)\sim\mu^{\wor k}\\j_l\neq i}}\EE\limits_{0\leq k^\prime < k} I(X_i:X_{j_{k^\prime+1}}|X_{j_1}\dots X_{j_{k^\prime}})\\
  \nonumber
 &=\EE\limits_{i\sim\pi}\EE\limits_{0\leq k^\prime < k}\EE\limits_{\substack{(j_1,\dots,j_{k^\prime+1})\sim\mu^{\wor k^\prime+1}\\j_l\neq i}} I(X_i:X_{j_{k^\prime+1}}|X_{j_1}\dots X_{j_{k^\prime}}). 
  \nonumber&\text{(by \eqref{eqMarginalProp})}
   \nonumber
\end{align}
Note that we can bound 
\begin{align}
 \mu^{\wor k^\prime+1}(j_1,\dots,j_{k^\prime}) &= \frac{\mu(j_1)\mu(j_2)\dots\mu(j_{k^\prime+1})}{(1-\mu(j_1))\dots(1-\mu(j_1)-\dots-\mu(j_{k^\prime}))}\\
  \nonumber
 &\geq \mu(j_{k^\prime+1}) \frac{\mu(j_1)\dots\mu(j_{k^\prime})}{(1-\mu(j_1))\dots(1-\mu(j_1)-\dots-\mu(j_{k^\prime-1}))}\\
  \nonumber
 &=\mu(j_{k^\prime+1})\mu^{\wor k^\prime}(j_1,\dots,j_{k^\prime}). 
  \nonumber
\end{align}
Exploiting further that $I(X_i:X_j|X_{j_1}\dots X_{j_k^\prime})=0$ if $i\in \{j_l\}$ or $j\in \{j_l\}$,  we can decouple the summation over the first and last entry of the sampled tuples and obtain
\begin{align}
 \EE\limits_{i\sim\pi}I(X_i:X_{-i})  &\geq \EE\limits_{0\leq k^\prime < k}\sum\limits_{i\neq j}\sum\limits_{\substack{j_1,\dots,j_{k^\prime}\\\text{distinct}}}\pi(i)\mu(j)\mu^{\wor k^\prime}(j_1,\dots,j_{k^\prime})I(X_i:X_{j}|X_{j_1}\dots X_{j_{k^\prime}})
\end{align}
which corresponds up to relabeling to the claim.
Furthermore, as the average over $C$ and $a,b$ decouples, we can bound
\begin{align}
 \min\limits_{k^\prime < k,(c_1,\dots,c_{k^\prime})\subset[n]}& \EE\limits_{a\neq b}\pi(a)\mu(b) I(X_a:X_b|X_{c_1}\dots X_{c_{k^\prime}}) \nonumber\\
 &\qquad\leq \EE\limits_{k^\prime<k}\EE\limits_{(c_1,\dots,c_{k^\prime})}\EE\limits_{a\neq b}\pi(a)\mu(b) I(X_a:X_b|X_{c_1}\dots X_{c_{k^\prime}})
   \nonumber\\
 &\qquad\leq \frac{\ln(d)}{k}.
\end{align}
\end{proof}

Equipped with the
results of Lemma \ref{lm:decoupling_star} we can proof the star-graph case of Theorem~\ref{thm:special_distinguish_psa} as well as Theorem~\ref{thm:gen_distinguish_psa} along the lines of Ref.\ \cite{Brandao2016}.
Let us first proof the star graph case of Theorem~\ref{thm:special_distinguish_psa} in order to highlight the general structure and then proceede with the proof of Theorem~\ref{thm:gen_distinguish_psa} subsequently.

\subsection{Proof of star graph result in Theorem~\ref{thm:special_distinguish_psa}}\label{sec:spin_star}

For convenience, we repeat the statement of the theorem here.

\begin{customthm}{10}[Product state approximation for distinguishable particles on star graphs]
For $(V,E)$ being a star in which $\syssize-1$ leaves couple to one root it follows that for any $\rho\in\mathcal{D}((\C^d)^{\otimes \syssize})$ there exists a globally separable state $\sigma$ such that 
 \begin{equation}
  \EE\limits_{(i,j)\sim E}\|\rho^{i,j}-\sigma^{i,j}\|_1 \leq 18 \left(\frac{d^2\ln(d)}{\syssize-1}\right)^{\frac{1}{3}}.
 \end{equation} 
\end{customthm}
\begin{proof}
 By Lemma 14 of Ref.\ \cite{Brandao2016}, for any dimension $d$ of the local Hilbert space and integer $k$ there is an informationally complete measurement $\Lambda$ with $\leq d^8$ outcomes such that for any trace-less operator $\xi\in\mathcal{B}((\C^d)^{\otimes k})$ we have 
\begin{equation}
 \frac{1}{(18d)^{k/2}}\|\xi\|_1 \leq \|\Lambda^{\otimes k}(\xi)\|_1.
\end{equation}
Define the probability distributions $\mu$ and $\pi$ over $[\syssize]$ via $\mu(i) = \delta_{i\in[N-1]}/(N-1)$, i.e., being the uniform distribution over the first $N-1$ nodes and set $\pi(i) = \delta_{i,N}$.

 Let $\rho^{[\syssize]}=\rho$ be a state on all $\syssize$ systems with local dimension $d$. Let $p^{[\syssize]}=\Lambda^{\otimes \syssize}(\rho)$ be the probability distribution obtained by measuring all $\syssize$ subsystems using the local informationally complete measurement $\Lambda$. For any $k$-tuple $C=(c_1,\dots,c_k)\subset [\syssize-1]$ define $\tau_C$ to be the state after measuring the subsystems labeled by $C$ and denote with $C^C$ the remaining system. Let $p^{C}$ be the outcome distribution of the measurement on $C$, we can write 
 \begin{equation}
  \tau_C = \EE\limits_{x\sim p^{C}}\tau_{C,x}^{C^C}\otimes \ketbra{x}{x}^{C} .
 \end{equation}
We then extend any $\tau_{C,x}^{C^C}$ trivially to the full system again by defining 
 \begin{equation}
  \tau_{C,x}^{[\syssize]} := \tau_{C,x}^{C^C}\otimes (\id/d)^{\otimes k}.
 \end{equation}
We define the product state 
\begin{align}
 \sigma^{[\syssize]}_{C,x} := \tau^{1}_{C,x}\otimes\dots\otimes \tau^{\syssize}_{C,x}
\end{align}
and correspondingly the separable state
\begin{align}
 \sigma^{[\syssize]}_{C} := \EE\limits_{x\sim p^{C}}\sigma^{[\syssize]}_{C,x}.
\end{align}
 Note that $\sigma^{[\syssize]}_{C,x}$ and $ \tau^{[\syssize]}_{C,x}$ agree trivially on $C^C$.
 Our aim now is to bound
 \begin{equation}
  \epsilon_C = \sum\limits_{i\neq j}\pi(i)\mu(j)\|\rho^{i,j}-\sigma^{i,j}_C\|_1
 \end{equation}
for some (unknown) $C$.
We start with
\begin{align}
 \epsilon_C &= \sum\limits_{i\neq j}\pi(i)\mu(j)\|\rho^{i,j}-\sigma^{i,j}_C\|_1\\
 \nonumber
 &=\sum\limits_{i\neq j}\pi(i)\mu(j)\|\rho^{i,j}-\tau^{i,j}_C+\tau^{i,j}_C-\sigma^{i,j}_C\|_1\\
  \nonumber
 &\leq \sum\limits_{i\neq j}\pi(i)\mu(j)\left(\|\rho^{i,j}-\tau^{i,j}_C\|_1+\|\tau^{i,j}_C-\sigma^{i,j}_C\|_1\right).
  \nonumber
 \end{align}
 Splitting the sum into contributions where $i\in C$ or $j\in C$ and $i,j\in C^C$ we obtain by exploiting $\|\rho-\sigma\|_1\leq 2$ and that $\rho$ and $\tau_C$ agree on $C^C$ as well as inserting the definition of $\tau_C$ and $\sigma_C$ yields
 \begin{align}
 \epsilon_C &\leq 2\sum\limits_{i\in C,j\neq i}\pi(i)\mu(j)+2\sum\limits_{j\in C,i\neq j}\pi(i)\mu(j)+\sum\limits_{i\neq j}\pi(i)\mu(j)\|\EE\limits_{x\sim p^{C}}(\tau_{C,x}^{i,j}-\sigma_{C,x}^{i,j})\|_1\\
  \nonumber
 &\leq 2\sum\limits_{i\in C,j\neq i}\pi(i)\mu(j)+2\sum\limits_{j\in C,i\neq j}\pi(i)\mu(j)+\EE\limits_{x\sim p^{C}}\sum\limits_{i\neq j}\pi(i)\mu(j)\|\tau_{C,x}^{i,j}-\tau_{C,x}^{i}\otimes \tau^{j}_{C,x}\|_1.
  \nonumber
\end{align}
Note that $\delta_{i\in C}\pi(i)$ is always zero. In addition, we have that $\sum\limits_{j}\delta_{j\in C}\mu(j) \leq k/(\syssize-1)$, hence, 
\begin{align}
 \epsilon_C &\leq\frac{2k}{\syssize-1}+\EE\limits_{x\sim p^{C}}\sum\limits_{i\neq j}\pi(i)\mu(j)\|\tau_{C,x}^{i,j}-\tau_{C,x}^{i}\otimes \tau^{j}_{C,x}\|_1.
\end{align}
Due to the informationally completeness of the measurement, 
we obtain that 
\begin{align}
 \epsilon_C &\leq\frac{2k}{\syssize-1}+\EE\limits_{x\sim p^{C}}\sum\limits_{i\neq j}\pi(i)\mu(j)18 d \|p_{x}^{i,j}-p_{x}^{i}\otimes p^{j}_{x}\|_1.
\end{align}
Pinsker's inequality reads then
\begin{align}
 \epsilon_C &\leq\frac{2k}{\syssize-1}+\EE\limits_{x\sim p^{C}}\sum\limits_{i\neq j}\pi(i)\mu(j)18 d \sqrt{2 I(i:j)_{p_{x}}}.
\end{align}
Exploiting the convexity of $x^2$ yields 
\begin{align}
 \epsilon_C &\leq\frac{2k}{\syssize-1}+\sum\limits_{i\neq j}\pi(i)\mu(j)18 d \sqrt{2 I(i:j|C)}.
\end{align}
Choose now $C_0$ according to Lemma \ref{lm:decoupling_star} such that 
\begin{align}
 \epsilon_{C_0} &\leq\frac{2k}{\syssize-1}+18 d \sqrt{2\frac{\ln(d)}{k}}.
\end{align}
Balancing the two terms with $k = (9^2   d^2 2\ln(d)(\syssize-1)^2)^{1/3}$ yields
\begin{align}
 \epsilon_{C_0} &\leq4\left(\frac{9^2 d^2 2\ln(d)}{\syssize-1}\right)^{\frac{1}{3}}\\
 \nonumber
    &\leq22\left(\frac{d^2 \ln(d)}{\syssize-1}\right)^{\frac{1}{3}}.
\end{align}

\end{proof}

\subsection{Proof of Theorem~\ref{thm:gen_distinguish_psa}}
Again, for convenience, we repeat again the statement of the theorem:

\begin{customthm}{8}[Product state approximation]
 Let $G$ be a symmetric matrix with positive entries and $\sum_{i,j\in V}G_{i,j} = 1$. Define $\pi_j = \sum_{i\in V} G_{i,j}$ and $A_{i,j} = G_{i,j}/\pi_j$. 
 Then, for any $\rho\in\mathcal{D}((\C^d)^{\otimes n})$ there exists a globally separable state $\sigma$ such that 
 \begin{equation}
  \EE\limits_{(i,j)\sim G}\|\rho^{i,j}-\sigma^{i,j}\|_1 \leq 47\left(d^4 \ln(d)\tr(A^2)\|\pi\|_2^2\right)^{1/5}+2\|\pi\|_2^2 .
 \end{equation} 
\end{customthm}

\begin{proof}
As in Section~\ref{sec:spin_star} we denote by $\Lambda$ an informationally complete measurement, mapping two distinct quantum states $\rho_1,\rho_2\in\mathcal{D}(\C^d)$ to two distinct classical probability distributions. By Lemma 14 of
Ref.\ \cite{Brandao2016}, for any dimension $d$ of the local Hilbert space and integer $k$,  there is an informationally complete measurement with $\leq d^8$ outcomes such that for any trace-less operator $\xi\in\mathcal{B}((\C^d)^{\otimes k})$ we have 
\begin{equation}
 \frac{1}{(18d)^{k/2}}\|\xi\|_1 \leq \|\Lambda^{\otimes k}(\xi)\|_1.
\end{equation}

 Let $\rho^{[\syssize]}$ be a state on $\syssize$ systems with local dimension $d$. Let $p^{[\syssize]}=\Lambda^{\otimes \syssize}(\rho)$ be the probability distribution obtained by measuring all $\syssize$ subsystems using the local informationally complete measurement $\Lambda$. For 
 any $k$-tuple $C=(c_1,\dots,c_k)$, 
 define $\tau_C$ to be the state after measuring the 
 subsystems labeled by $C$ and denote with $C^C$ to 
 the remaining system. Let $p^{C}$ be the outcome 
 distribution of the measurement on $C$ we can write 
 \begin{equation}
  \tau_C = \EE\limits_{x\sim p^{X_C}}\tau_{C,x}^{C^C}\otimes \ketbra{x}{x}^{X_C}.
 \end{equation}
Extend then any $\tau_{C,x}^{C^C}$ trivially to the full system again by defining 
 \begin{equation}
  \tau_{C,x}^{[\syssize]} = \tau_{C,x}^{C^C}\otimes (\id/d)^{\otimes k}.
 \end{equation}
Define the product state 
\begin{align}
 \sigma^{[\syssize]}_{C,x} = \tau^{1}_{C,x}\otimes\dots\otimes \tau^{\syssize}_{C,x}
\end{align}
and correspondingly the separable state
\begin{align}
 \sigma^{[\syssize]}_{C} = \EE\limits_{x\sim p^{X_C}}\sigma^{[\syssize]}_{C,x}.
\end{align}
 Note that $\sigma^{[\syssize]}_{C,x}$ and $ \tau^{[\syssize]}_{C,x}$ agree trivially on $C^C$.
 Our aim now is to bound
 \begin{equation}
  \epsilon_C = \sum\limits_{i\neq j}G_{i,j}\|\rho^{i,j}-\sigma^{i,j}_C\|_1
 \end{equation}
for some (unknown) $C$.
We start with
\begin{align}
 \epsilon_C &= \sum\limits_{i\neq j}G_{i,j}\|\rho^{i,j}-\sigma^{i,j}_C\|_1
 \\
 \nonumber
 &=\sum\limits_{i\neq j}G_{i,j}\|\rho^{i,j}-\tau^{i,j}_C+\tau^{i,j}_C-\sigma^{i,j}_C\|_1\\
 \nonumber
 &\leq \sum\limits_{i\neq j}G_{i,j}\left(\|\rho^{i,j}-\tau^{i,j}_C\|_1+\|\tau^{i,j}_C-\sigma^{i,j}_C\|_1\right) .
  \nonumber
 \end{align}
 Splitting the sum into contributions where $i\in C$ or $j\in C$ and $i,j\in C^C$ we obtain by exploiting $\|\rho-\sigma\|_1\leq 2$ and that $\rho$ and $\tau_C$ agree on $C^C$ as well as inserting the definition of $\tau_C$ and $\sigma_C$ yields
 \begin{align}
 \epsilon_C &\leq 2\sum\limits_{i\in C,j\neq i}G_{i,j}+2\sum\limits_{j\in C,i\neq j}G_{i,j}+\sum\limits_{i\neq j}G_{i,j}\|\EE\limits_{x\sim p^{X_C}}(\tau_{C,x}^{i,j}-\sigma_{C,x}^{i,j})\|_1\\
  \nonumber
 &\leq 2\sum\limits_{i\in C,j\neq i}G_{i,j}+2\sum\limits_{j\in C,i\neq j}G_{i,j}+\EE\limits_{x\sim p^{X_C}}\sum\limits_{i\neq j}G_{i,j}\|\tau_{C,x}^{i,j}-\tau_{C,x}^{i}\otimes \tau^{j}_{C,x}\|_1 .
\end{align}
Let us define 
\begin{align}
2\sum\limits_{i\in C,j\neq i}G_{i,j}+2\sum\limits_{j\in C,i\neq j}G_{i,j} = :P(C) 
\end{align}
and deal with it later.
By the informationally completeness of the measurement we conclude that 
\begin{align}
 \|\tau_{C,x}^{i,j}-\tau_{C,x}^{i}\otimes \tau^{j}_{C,x}\|_1 \leq 18 d \|p_{x}^{i,j}-p_{x}^{i}\otimes p^{j}_{x}\|_1.
\end{align}
For $i\neq j$ define 
 \begin{equation}
 \Delta_{C,x}(i,j) = \|p_{x}^{i,j}-p_{x}^{i}\otimes p^{j}_{x}\|_1.
\end{equation}
Pinsker's inequality reads then 
\begin{equation}
 \Delta_{C,x}(i,j)\leq \sqrt{2 I(i:j)_{p_x}},
\end{equation}
and this inequality immediately yields
\begin{align}
 \EE\limits_{x\sim p^{X_C}}\sum\limits_{i\neq j}\pi(i)\pi(j)\Delta_{C,x}(i,j)^2  &\leq 2 \sum\limits_{i\neq j}\pi(i)\pi(j)\sum\limits_{x}p^{X_C}(x)I(i:j)_{p_x}\\
 &= 2 \sum\limits_{i\neq j}\pi(i)\pi(j)I(i:j|C).
  \nonumber
 \end{align}
 We define the shorthand notations $\Delta_C := \EE\limits_{x\sim p^{X_C}}\Delta_{C,x}$ and  
 \begin{align}
 \gamma_C := \sum\limits_{i\neq j}\pi(i)\pi(j)I(i:j|C).
  \end{align}
 Then applying the self-decoupling 
 Lemma \ref{lm:decoupling_star} one finds
 \begin{equation}
  \EE\limits_C \gamma_C = \EE\limits_{0\leq k^\prime < k}\EE\limits_{c_1,\dots,c_{k^\prime}\sim\pi^{\wor k^\prime}}\sum\limits_{i\neq j}\pi(i)\pi(j)I(i:j|\{c_1,\dots,c_{k^\prime}\}) \leq \frac{\ln(d)}{k}.
 \end{equation}
 In order to bound $\epsilon_C$ we would like to bound the average $\EE_{i,j\sim G}\Delta_{C}(i,j)$ but only know how to bound the average according to $\pi$ by the decoupling lemma. In order to get from the one bound to the other define
 \begin{equation}
  B_\lambda = \{(i,j):G_{i,j}>\lambda \pi(i)\pi(j)\} = \{(i,j):A_{i,j}>\lambda \pi(i)\}.
 \end{equation}
The weight of the set $B_\lambda$ can then be bound as
\begin{align}
 \eta_\lambda &= \sum\limits_{(i,j)\in B}G_{i,j} = \sum\limits_{(i,j)\in B}A_{i,j}\pi(j) \leq \sum\limits_{(i,j)\in B}A_{i,j}\pi(j) \leq \sum\limits_{(i,j)\in B}\frac{A_{i,j}A_{j,i} }{\lambda} \leq \frac{\tr(A^2)}{\lambda}  .
\end{align}
With a bound on the total weight in $B$, we know there is a weight matrix $G^\prime$ such that we obtain elementwise 
\begin{equation}
 G_{i,j} \leq \lambda \pi(i)\pi(j) + \eta_\lambda G^\prime_{i,j}. 
\end{equation}
With this, we obtain from the Cauchy Schwarz 
inequality
\begin{align}
 \sum\limits_{i\neq j}G_{i,j}\Delta_C(i,j) &= \sum\limits_{i \neq j}\sqrt{G_{i,j}}\sqrt{G_{i,j}}\Delta_C(i,j) \\
  \nonumber
 &\leq \sqrt{\sum\limits_{i\neq j}G_{i,j}\Delta_C(i,j)^2} \\
  \nonumber
 &\leq \sqrt{\sum\limits_{i\neq j}\lambda\pi(i)\pi(j)\Delta_C(i,j)^2 + \sum\limits_{i\neq j}\eta G_{i,j}^\prime\Delta_C(i,j)^2} \\
  \nonumber
 &\leq \sqrt{\sum\limits_{i\neq j}\lambda\pi(i)\pi(j)\EE\limits_{x\sim p^{X_C}}\Delta_{C,x}(i,j)^2 + 4\eta_\lambda} \\
  \nonumber
 &\leq \sqrt{2\lambda\gamma_C + 4\eta_\lambda} \\
  \nonumber
 &\leq \sqrt{2\lambda\gamma_C + 4\frac{\tr(A^2)}
  \nonumber{\lambda}} .
\end{align}
Choosing $\lambda: = \sqrt{2\tr(A^2)/\gamma_C}$ such that both error terms are balanced, we find 
\begin{equation}
 \sum\limits_{i,j}G_{i,j}\Delta_C(i,j) \leq \left(32\gamma_C\tr(A^2)\right)^{1/4} .
\end{equation}
Inserting this result above yields
\begin{equation}
 \epsilon_C \leq P(C) + 18 d\left(32\gamma_C\tr(A^2)\right)^{1/4}.
\end{equation}
Now choose $C_0$ such that the right hand side of the bound is minimal. 
We then obtain  
\begin{align}
 \epsilon_{C_0} &\leq P(C_0) + 18 d\left(32\gamma_{C_0}\tr(A^2)\right)^{1/4}\nonumber\\
 &\leq\EE\limits_{C}P(C) + \EE\limits_{C} 18 d\left(32\gamma_{C}\tr(A^2)\right)^{1/4}\nonumber\\
 &\leq\EE\limits_{C}P(C) +  18 d\left(32\EE\limits_{C}\gamma_{C}\tr(A^2)\right)^{1/4}\label{eq:app:beforePaul}
\end{align}
as the minimum is smaller than the average and taking the fourth root is concave.

As we know that  $\EE_C\gamma_C \leq \ln(d)/k$, all we need is a bound for $\EE_C P(C)$ which can be obtained as follows \cite{PaulPrivate}.
Let us denote by $p_k(c)$ the probability that $c\in[n]$ is contained in a draw according to $\pi^{\wor k}$. Note that $\sum_c p_k(c) = k$. We can then write
\begin{align}
 \EE_C P(C) &= \frac{4}{k}\sum\limits_{0\leq k^\prime < k}\sum\limits_{\substack{c_1,\dots,c_{k^\prime}\\\text{distinct}}}\sum\limits_{l=1,\dots,k^\prime}
 \pi^{\wor k^\prime}(c_1,\dots,c_{k^\prime})\pi(c_l)\\
 &=\frac{4}{k}\sum\limits_{0\leq k^\prime < k}\sum\limits_{c}p_{k^\prime}(c)\pi(c).\nonumber
 \end{align}
According to Theorem 4.1 in Ref.~\cite{Kochar2001} it holds for any $k\leq n$
\begin{align}
 (p_{k^\prime}(c))_{c\in[n]} \prec k^\prime(\pi(c))_{c\in[n]}
\end{align}
such that we can conclude that 
\begin{align}
 \sum\limits_{c}p_k(c)\pi(c) \leq \max\limits_{\sigma\in S_n}\sum\limits_{c}p_k(\sigma(c))\pi(c) \leq k\max\limits_{\sigma\in S_n}\sum\limits_{c}\pi(\sigma(c))\pi(c) = k\sum\limits_c \pi(c)^2 = k\|\pi\|_2^2,
\end{align}
where the second inequality holds as $\pi(c)\geq0$ and $p_k(c)\geq0$ for all $c$.
From this we conclude that 
\begin{align}
 \EE_C P(C) \leq 2k\|\pi\|_2^2.
\end{align}
Inserting this bound into Eq.~(\ref{eq:app:beforePaul}), we obtain
\begin{align}
 \epsilon_{C_0} \leq 2k\|\pi\|_2^2 + 18d\left(\frac{32\tr(A^2)\ln(d)}{k}\right)^{1/4}.
\end{align}
Noting that $k= (9d)^{4/5}(32\tr(A^2)\ln(d)/\|\pi\|_2^8)^{1/5}$ would balance the two terms and round it to the next largest integer would increase the error not more than $2\|\pi\|_2^2$ yields
\begin{align}
 \epsilon_{C_0} &\leq 2\left[(18 d)^4 64\ln(d)\tr(A^2)\|\pi\|_2^2\right]^{1/5}+2\|\pi\|_2^2\nonumber \\
 &\leq 47\left[d^4 \ln(d)\tr(A^2)\|\pi\|_2^2\right]^{1/5}+2\|\pi\|_2^2.\label{eq:app:bound}
\end{align}
Note that the change in the pre-factor and exponent of the first term \eqref{eq:app:bound} with respect to the result of Ref.~\cite{Brandao2016} results from a different expression for $P(C)$ (Ref.~\cite{Brandao2016} appears to neglect one possible contribution and misses a factor of 2) as well as a different choice of $k$ in the last bound (the evaluation in Ref.~\cite{Brandao2016} seems erroneous here).  
\end{proof}

\subsection{Additional special case: distinguishable particles on complete bipartite graphs}
\label{sec:Special case: complete bipartite graphs}
In Ref.~\cite{Brandao2016} next to the general case, the above regular graphs are discussed as special cases. 
Here we want to add to the star case discussed in Theorem \ref{thm:special_distinguish_psa} the case of a complete bipartite graph with $V = [n] = A\cup B$ with $A\cap B = \emptyset$ and $E = \{(i,j):i\in A,j\in B\}$.


\begin{thm}[Product state approximation on complete bipartite graphs]
 Let $V=[n] = A\cup B$ with $A,B$ disjoint and $E= \{(i,j):i\in A,j\in B\}$ be the edge set of the complete bipartite graph in which every site in $A$ is connected to every site in $B$ without any internal edges in $A$ or $B$.
 Then, for any $\rho\in\mathcal{D}((\C^d)^{\otimes n})$ there exists a globally separable state $\sigma$ such that 
 \begin{equation}
  \EE\limits_{j\sim [n]}\|\rho^{i,j}-\sigma^{i,j}\| \leq  \leq22\left(\frac{d^2 \ln(d)}{n}\right)^{\frac{1}{3}}.
 \end{equation} 
\end{thm}

\begin{proof}
As in Ref.~\cite{Brandao2016}, we denote by $\Lambda$ an informationally complete measurement, mapping two distinct quantum states $\rho_1,\rho_2\in\mathcal{D}(\C^d)$ to two distinct classical probability distributions. By Lemma 14 of Ref.\ \cite{Brandao2016} for any dimension $d$ of the local Hilbert space and integer $k$ there is an informationally complete measurement with $\leq d^8$ outcomes such that for any trace-less operator $\xi\in\mathcal{B}((\C^d)^{\otimes k})$ we have 
\begin{equation}
 \frac{1}{(18d)^{k/2}}\|\xi\|_1 \leq \|\Lambda^{\otimes k}(\xi)\|_1.
\end{equation}
Similar to the general case discussed above, we condition the probability distribution which results from applying $\Lambda$ to a given state on a subset of its constituents. However, in view of Lemma \ref{lm:decoupling_star}, we sample the set $C$ on which we condition the state only from  $B$ and sample the first index $a$ from $A$. 

 Let $\rho^Q$ be a state on $n=|V|$ systems with local dimension $d$ while $Q=(Q_1,\dots,Q_{n})$ denotes the full system. Let $p^{X_1\dots X_{n}}=\Lambda^{\otimes n}(\rho)$ be the probability distribution obtained by measuring all $n$ subsystems using the local informationally complete measurement $\Lambda$. For any $k$-tuple $C=(c_1,\dots,c_k)\subset B$ define $\tau_C$ to be the state after measuring the subsystems labeled by $C$ and define $Q^\prime$ to be the remaining system. Let $p^{X_C}$ be the outcome distribution of the measurement on $C$ we can write 
 \begin{equation}
  \tau_C = \EE\limits_{x\sim p^{X_C}}\tau_{C,x}^{Q^\prime}\otimes \ketbra{x}{x}^{X_C}.
 \end{equation}
We then extend any $\tau_{C,x}^{Q^\prime}$ trivially to the full system again by defining 
 \begin{equation}
  \tau_{C,x}^Q = \tau_{C,x}^{Q^\prime}\otimes (\id/d)^{\otimes k}.
 \end{equation}
Define the product state 
\begin{align}
 \sigma^Q_{C,x} = \tau^{Q_1}_{C,x}\otimes\dots\otimes \tau^{Q_{n}}_{C,x}
\end{align}
and correspondingly the separable state
\begin{align}
 \sigma^Q_{C} = \EE\limits_{x\sim p^{X_C}}\sigma^Q_{C,x}.
\end{align}
 Note that $\sigma^Q_{C,x}$ and $ \tau^Q_{C,x}$ agree trivially on $Q_C$.
 Our aim now is to bound
 \begin{equation}
  \epsilon_C = \sum\limits_{a\neq b}\pi(a)\mu(b)\|\rho^{Q_a,Q_b}-\sigma^{Q_a,Q_b}_C\|_1
 \end{equation}
for some (unknown) $C$ with $\pi(a) = \delta_{a\in A}/|A|$ and $\mu(b) = \delta_{b\in B}/|B|$ being the uniform distributions over corresponding set.
We start with
\begin{align}
 \epsilon_C &= \sum\limits_{a\neq b}\pi(a)\mu(b)\|\rho^{Q_a,Q_b}-\sigma^{Q_a,Q_b}_C\|_1\\
 \nonumber
 &=\sum\limits_{a\neq b}\pi(a)\mu(b)\|\rho^{Q_a,Q_b}-\tau^{Q_a,Q_b}_C+\tau^{Q_a,Q_b}_C-\sigma^{Q_a,Q_b}_C\|_1\\
  \nonumber
 &\leq \sum\limits_{a\neq b}\pi(a)\mu(b)\left(\|\rho^{Q_a,Q_b}-\tau^{Q_a,Q_b}_C\|_1+\|\tau^{Q_a,Q_b}_C-\sigma^{Q_a,Q_b}_C\|_1\right) .
  \nonumber
 \end{align}
 Splitting the sum into contributions where $a\in C$ or $b\in C$ and $a,b\in C^C$ we obtain by exploiting $\|\rho-\sigma\|_1\leq 2$ and that $\rho$ and $\tau_C$ agree on $C^C$ as well as inserting the definition of $\tau_C$ and $\sigma_C$ yields
 \begin{align}
 \epsilon_C &\leq 2\sum\limits_{a\in C,b\neq a}\pi(a)\mu(b)+2\sum\limits_{b\in C,a\neq b}\pi(a)\mu(b)+\sum\limits_{a\neq b}\pi(a)\mu(b)\|\EE\limits_{x\sim p^{X_C}}(\tau_{C,x}^{Q_a,Q_b}-\sigma_{C,x}^{Q_a,Q_b})\|_1\\
 \nonumber
 &\leq 2\sum\limits_{a\in C,b\neq a}\pi(a)\mu(b)+2\sum\limits_{b\in C,a\neq b}\pi(a)\mu(b)+\EE\limits_{x\sim p^{X_C}}\sum\limits_{a\neq b}\pi(a)\mu(b)\|\tau_{C,x}^{Q_a,Q_b}-\tau_{C,x}^{Q_a}\otimes \tau^{Q_b}_{C,x}\|_1.
\end{align}
Note that $\delta_{a\in C}\pi(a)$ is always zero. In addition, we have that $\sum\limits_{b}\delta_{b\in C}\mu(b) \leq k/|B|$ hence, 
\begin{align}
 \epsilon_C &\leq \frac{2k}{|B|}+\EE\limits_{x\sim p^{X_C}}\sum\limits_{a\neq b}\pi(a)\mu(b)\|\tau_{C,x}^{Q_a,Q_b}-\tau_{C,x}^{Q_a}\otimes \tau^{Q_b}_{C,x}\|_1.
\end{align}
Due to the informationally completeness of the measurement we obtain that 
\begin{align}
 \epsilon_C &\leq\frac{2k}{n}+\EE\limits_{x\sim p^{X_C}}\sum\limits_{a\neq b}\pi(a)\mu(b)18 d \|p_{x}^{X_a,X_b}-p_{x}^{X_a}\otimes p^{X_b}_{x}\|_1.
\end{align}
Pinsker's inequality reads then
\begin{align}
 \epsilon_C &\leq\frac{2k}{n}+\EE\limits_{x\sim p^{X_C}}\sum\limits_{a\neq b}\pi(a)\mu(b)18 d \sqrt{2 I(X_a:X_b)_{p_{x}}}.
\end{align}
Exploiting the convexity of $x^2$ yields 
\begin{align}
 \epsilon_C &\leq\frac{2k}{n}+\sum\limits_{a\neq b}\pi(a)\mu(b)18 d \sqrt{2 I(X_a:X_b|X_C)}.
\end{align}
Choose now $C_0$ according to Lemma \ref{lm:decoupling_star} such that 
\begin{align}
 \epsilon_{C_0} &\leq\frac{2k}{n}+18 d \sqrt{2\frac{\ln(d)}{k}}.
\end{align}
Balancing the two terms with $k = (9^2   d^2 2\ln(d)n^2)^{1/3}$ yields
\begin{align}
 \epsilon_{C_0} &\leq4\left(\frac{9^2 d^2 2\ln(d)}{n}\right)^{\frac{1}{3}}\\
    &\leq22\left(\frac{d^2 \ln(d)}{n}\right)^{\frac{1}{3}},\nonumber
\end{align}
which gives the statement to be shown.
\end{proof}

\end{document}